\documentclass[12pt,nofootinbib,tightenlines,longbibliography,onecolumn,notitlepage%twocolumn%
]{revtex4-1}

\usepackage{tikz}
\usetikzlibrary{matrix,arrows,shapes}
\usepackage{pgf}
\usepackage{cancel}
\usepackage{url}
\usepackage{longtable}
\usepackage{enumitem}
\usepackage{amsthm}
\usepackage{amssymb}
\usepackage{amsmath}
\usepackage{amscd}
\usepackage{eucal}
\usepackage{mathrsfs}
\usepackage{upgreek}
\usepackage{mathtools}
\usepackage{dsfont}
\usepackage{yfonts}
\usepackage[T1]{fontenc}
\usepackage{bbm}
\usepackage{array}
\usepackage{booktabs}
\usepackage{feynmf}
\usetikzlibrary{trees}
\usetikzlibrary{decorations.pathmorphing}
\usetikzlibrary{decorations.markings}
\tikzset{
	photon/.style={decorate, decoration={snake}, draw=red},
	electron/.style={draw=blue, postaction={decorate},
		decoration={markings,mark=at position .55 with {\arrow[draw=blue]{>}}}},
	gluon/.style={decorate, draw=magenta,
		decoration={coil,amplitude=4pt, segment length=5pt}},
	sderiv/.style={postaction={decorate},
		decoration={markings,mark=at position .3 with {\arrow{>}}}},
	tderiv/.style={postaction={decorate},
		decoration={markings,mark=at position .7 with {\arrow{<}}}},
	stderiv/.style={postaction={decorate},
		decoration={markings,mark=at position .7 with {\arrow{<}},mark=at position .3 with {\arrow{>}}}}
}
\usepackage{scalerel}
\usepackage{stackengine}

\definecolor{see}{RGB}{67,75,179}
\definecolor{darksee}{RGB}{42,44,148}
\definecolor{honey}{RGB}{232,180,129}
\definecolor{lighthoney}{RGB}{255,254,220}
\definecolor{citecol}{rgb}{0.5,0,0} 
\definecolor{blue1}{RGB}{130,150,209}
\usepackage{soul}
\usepackage[colorinlistoftodos,textsize=footnotesize]{todonotes}

\setstcolor{red}
\definecolor{see}{RGB}{67,75,179}
\usepackage{hyperref}
\hypersetup
{
	bookmarks=true,
	bookmarksopen=true,
	pdfmenubar=true, %menubar shown
	pdfhighlight=/O, %effect of clicking on a link
	colorlinks=true, %couleurs sur les liens hypertextes
	pdfpagemode=None, %aucun mode de page
	pdfpagelayout=SinglePage, %ouverture en simple page
	pdffitwindow=true, %pages ouvertes entierement dans toute la fenetre
	linkcolor=see, %couleur des liens hypertextes internes
	citecolor=red!75!darkgray, %couleur des liens pour les citations
	urlcolor=see %couleur des liens pour les url
}

\newcommand{\no}[1]{\mathop{\mathopen: {#1} \mathclose:}}
\newcommand{\fA}{\mathfrak{A}}

\newcommand{\fP}{\mathfrak{P}}

\newcommand{\HH}{\mathcal{H}}
\newcommand{\Gcal}{\mathcal{G}}  % gauge group
\newcommand{\Lcal}{\mathcal {L}}
\newcommand{\Hcal}{\mathcal {H}}

\newcommand{\Kcal}{\mathcal{K}}  
\newcommand{\Ccal}{\mathcal{C}}
\newcommand{\Dcal}{\mathcal{D}}
\newcommand{\Ecal}{\mathcal{E}} 

\newcommand{\Fcal}{\mathcal{F}}

\newcommand{\Mcal}{\mathcal{M}}

\newcommand{\Scal}{\mathcal{S}}

\newcommand{\Tcal}{\mathcal{T}}

\newcommand{\Wcal}{\mathcal{W}}
\newcommand{\Xcal}{\mathcal{X}}

\newcommand{\Ci}{\mathcal{C}^\infty} % smooth functions
\newcommand{\link}{\prec\!\!*\,\,}
    
\newcommand{\loc}{\mathrm{loc}}

\newcommand{\reg}{\mathrm{reg}}
\newcommand{\pol}{\mathrm{pol}}

\newcommand{\vx}{\underline{x}}
\newcommand{\vy}{\underline{y}}
\newcommand{\NN}{\mathbb{N}}          % natural naumbers
\newcommand{\ZZ}{\mathbbmss{Z}}     % Menge der ganzen Zahlen
\newcommand{\RR}{\mathbb{R}}           % real  numbers
\newcommand{\CC}{\mathbb{C}}           % complex numbers

\newcommand{\M}{\mathbb{M}} 	     % Minkowski spacetime
\newcommand{\al}{\alpha}

\newcommand{\la}{\lambda}

\newcommand{\ph}{\varphi}

\newcommand{\T}{\cdot_{{}^\Tcal}}

\newcommand{\TT}{\Tcal}

\newcommand{\sst}[1]{\scriptscriptstyle{#1}}  % small font for the subscripts
\newcommand{\1}{\mathds{1}}                         % identity
\newcommand{\be}{\begin{equation}}
\newcommand{\ee}{\end{equation}}

\DeclareMathOperator{\supp}{supp}      % support
\def\normOrd#1{\mathop{:}\nolimits\!#1\!\mathop{:}\nolimits}

\DeclareMathOperator{\rk}{rk}

\newcommand{\starb}{\mathbin{\star_{\sst \TT}}}
\DeclareMathOperator{\Aut}{Aut}
\newcommand{\abs}[1]{\lvert#1\rvert}
\newcommand{\acts}[1]{\overset{\twoheadrightarrow}{#1}}

\newcommand{\starbint}{\star_{\sst{\TT,\mathrm{int}}}}

\newcommand{\starHint}{\star_{\sst{H,\mathrm{int}}}}

\DeclareMathOperator*{\mean}{Mean}
\theoremstyle{plain}
\newtheorem{thm}{Theorem}[section]
\newtheorem{prop}[thm]{Proposition}
\newtheorem{cor}[thm]{Corollary}
\newtheorem{lemma}[thm]{Lemma}

\theoremstyle{definition}
\newtheorem{df}[thm]{Definition}
\theoremstyle{remark}
\newtheorem{rem}[thm]{Remark}

\newtheorem{exa}[thm]{Example}

\begin{document}		
	
	\title{Algebraic Classical and Quantum Field Theory  on Causal Sets}
	\author{Edmund Dable-Heath\thanks{Current address: Department of Electrical and Electronic Engineering, Imperial College London, South Kensington, London, SW7 2AZ, United Kingdom}}\email{ e.dable-heath18@imperial.ac.uk}
	\author{Christopher J.~Fewster}\email{chris.fewster@york.ac.uk}
	\author{Kasia Rejzner}\email{kasia.rejzner@york.ac.uk}
	\author{Nick Woods\thanks{Current address: TCM Group, Cavendish Laboratory, 19 JJ Thomson Avenue, Cambridge, CB3 0HE,  United Kingdom}}\email{nw361@cam.ac.uk}
	\affiliation{Department of Mathematics, University of York, Heslington, York YO10 5DD, United Kingdom}  
	\date{\today}
	\begin{abstract}
The framework of perturbative algebraic quantum field theory (pAQFT) is used to construct QFT models on causal sets. We discuss various discretised wave operators, including a new proposal based on the idea of a `preferred past', which we also introduce, and show how they may be used to construct classical free and interacting field theory models on a fixed causal set; additionally, we describe how the sensitivity of observables to changes in the background causal set may be encapsulated in a relative Cauchy evolution. These structures are used as the basis of a deformation quantization, using the methods of pAQFT. The SJ state is defined and discussed as a particular quantum state on the free quantum theory. Finally, using the framework of pAQFT, we construct interacting models for arbitrary interactions that are smooth functions of the field configurations. This is the first construction of such a wide class of models achieved in QFT on causal sets.
	\end{abstract}

	\maketitle	
	\tableofcontents

\section{Introduction}
Presently, our understanding of nature is split into two domains: one theory applies to quantum phenomena, and is relevant on the small scale; and a very different theory applies to gravity, space, and time, and is important for large-scale phenomena. Quantum gravity seeks to unify these two into one single description of nature. 
Whilst many attempts have been made, and various methods suggested, the problem of finding the unified theory of quantum gravity still remains open.  One of the fundamental conceptual problems that any such theory has to address is the understanding of the nature of space-time at small scales and the interplay of geometry and quantum phenomena.

This paper brings together two frameworks that have been used to develop theories that combine quantum effects and geometry. 
 The first, causal set theory \cite{Sorkin09,SorkinNotes,Henson09}, is based on the idea that the spacetime that we observe is not fundamental, but rather emergent from a discrete underlying structure. It is conjectured that in the small scale, spacetime is a discrete set of points and the only structure on this set is a partial order relation, interpreted as the causal structure.
 
  The second framework is that of algebraic quantum field theory (AQFT) \cite{HK,Haag} (see~\cite{FewRej2019} for a recent pedagogical introduction), and its generalization to curved spacetimes: locally covariant quantum field theory (LCQFT) \cite{BFV,HW} (see also \cite{FewVerchReview} for review) and perturbative algebraic quantum field theory (pAQFT) \cite{BF0,BDF,DF02,DFloop} (see also \cite{Book} for review). In LCQFT a model is defined by the assignment of 
topological $*$-algebras (often $C^*$-algebras) to globally hyperbolic spacetimes and algebra morphisms to causal embeddings of spacetimes. This assignment has to satisfy a number of axioms that generalize the Haag-Kastler axioms. In pAQFT, these topological *-algebras are formal power series in $\hbar$ and the coupling constant $\lambda$.

In this paper we apply LCQFT and pAQFT methods to QFT on causal sets. This brings benefit to both causal set theory and AQFT. In the first instance, the methods of pAQFT have been successfully applied to construct interacting QFT models in the continuum and now we use the same framework to construct interacting QFT models on causal sets. To our best knowledge, this is the first instance, where the general framework for introducing interaction in causal set theories has been proposed.

On the AQFT side, studying the discrete models allows one to avoid many of the technical difficulties related to UV divergences and study in detail the purely algebraic aspects of pAQFT and how the topology change affects LCQFT.

The main advantage of the algebraic framework is that many of the concepts used in the continuum translate very straightforwardly to the discrete case. For example, instead of assigning algebras to spacetimes, we assign algebras to causal sets. To follow the spirit of pAQFT, we start by defining the classical field theory on a causal set and then deform it using a simple formal deformation quantization prescription. The problem of defining classical dynamics on causal sets is, in our opinion, of interest on its own, since the usual canonical formalism does not apply in this situation. Instead, we use a variant of the Peierls prescription \cite{Pei} that allows us to introduce a Poisson bracket on the space of observables. We also show how to introduce interactions in this framework (following \cite{DF02}) using classical M{\o}ller operators. This is covered in section 3. For our constructions to work, we need to define, on a given causal set, the retarded (or advanced) Green function for the discretized field equation we consider. 
The retarded Green function is also a starting point in the approach of  \cite{Sorkin}.
We discuss various choices for discretization of the wave equation and for construction of Green functions. These include one \cite{Sorkin09,BDD10,DG13,Glaser_2014} which works well for sprinklings (locally finite subsets of Lorentzian manifolds, constructed by randomly selecting points from a given manifold using a Poisson distribution) \cite{Sorkin09} and the continuum limit is achieved by an averaging procedure.\footnote{Although the expectation value of the discretized field equation converges in mean, the variances diverge unless further nonlocal corrections are applied~\cite{Sorkin09}; see~\cite{Belenchia2015} for quantisations of such models.}  Another choice is based on an additional `preferred past structure', which we introduce in this work. It works well on a regular diamond lattice, for example.

After this paper was completed, our attention was drawn to the interesting paper \cite{FJ04} in which discrete d'Alembertians are formulated and the corresponding free theories quantised using the broad methodology of \cite{WaldBook}. The approach taken here is complementary in some respects: ref \cite{FJ04} is concerned with causal sets equipped with a slicing, which does not appear in our approach, but is essential to the definition of the symplectic form given in \cite{FJ04}. By contrast, we follow the spirit of Peierls covariant definition of the Poisson bracket, leading to a quantisation that can be applied to interacting theories. Another interesting contrast is that our use of `preferred past' structures for one of the discrete d'Alembertians considered, is much more local in nature than the global slicing structure of \cite{FJ04}. Nonetheless there are some very close parallels between the resulting discrete equations.

The idea of augmenting causal sets with some extra structure has a precedent in the works of  Cort{\^e}s and Smolin \cite{CS14,CS14a}, where elements of the causal set (events) carry momentum and energy, transmitted along causal links and conserved at each event. This is local in nature, but seems to be very different from our idea of augmenting the causal set with the `preferred past structure'. Nevertheless, it would be interesting to look for parallels between our approaches. 

In section 3.3, we discuss \textit{relative Cauchy evolution} (RCE) on causal sets. In \cite{BFV} the RCE was  introduced as the way to characterize the dynamics in LCQFT, see \cite{FewVerchReview} for further developments, and \cite{BFRej13} for an application to the characterization of background independence in perturbative quantum gravity. Relative Cauchy evolution measures the response of the dynamics to a local modification of the background spacetime (just as the stress-energy tensor in a continuum theory is obtained as a functional derivative of the action with respect to the metric). To define the RCE for causal sets, we first identify  distinguished regions, which we call past and future infinity, using the notion of \textit{layers} \cite{Sorkin09}. Then we consider two finite causal sets whose future and past infinity regions may be identified, so differences between the sets are localised in between. The RCE measures the response of the observables (classical or quantum) to that small change of the background causal set. In this work we study the RCE in the classical theory, but the generalization to quantum theory should be straightforward. We hope that RCE combined with ideas about dynamical generation of causal sets \cite{SR99} will allow us to understand how the evolution of observables on a causal set is related to the evolution of the causal set itself.

In section 4 we quantize the free theory using deformation quantization. In particular, we construct the Weyl algebra from the Poisson algebra of the classical theory and discuss states. We also show how to recover the Hilbert space representation of the Weyl algebra by considering the GNS representation. For the latter, one needs to fix a state and a possible choice in causal set theory is provided by
the Sorkin-Johnston (SJ) state \cite{SJ12,Johnston2}. This is a pure state which, as we emphasise, 
 is closely connected to a choice of inner product on the space of off-shell linear observables (in a finite causal set this is just $\RR^N$, $N\in\NN=\{1,2,\ldots\}$). The original SJ state is related in this way to the standard Euclidean inner product on $\RR^N$.  

However, if we want to take the continuum limit, it is better to modify the inner product on the space of linear observables, so that the state we obtain in the continuum is Hadamard. As shown
in \cite{FV12}, the continuum SJ state fails to be Hadamard on a large class of globally hyperbolic spacetimes (ultrastatic slabs).
It was later proven in \cite{FrancisThesis} that modifying the inner product on the space of smooth compactly-supported functions by means of changing the volume form on the underlying space-time results in the construction of a class of Hadamard states, interpreted as ``softened'' SJ states. This strategy for obtaining Hadamard states was first suggested by Sorkin in \cite{Sorkin}, as an alternative to the construction by Brum and Fredenhagen \cite{BrumF14}. The latter also produces a class of Hadamard states that can be interpreted as ``softened'' SJ states; an analogous construction for Dirac fields can be found in \cite{FL15}. (More discussion appears at the start of section~4.)

The results mentioned above suggest that one should be able to modify the inner product used for the construction of the state in the discrete setting, in such a way that the continuum limit would yield Hadamard state. We plan to follow this line of research in our future work.

 Last but not least, we close section 4 with the construction of the quantum interacting algebra $\fA_V(\Ccal)$, using the framework of pAQFT. As mentioned before, this result is of particular interest, since, to our best knowledge, this is the first systematic construction of interacting causal set quantum field theory models.

\section{Preliminaries}
\subsection{Causal sets}\label{sec:CausalSets}
A feature common to many quantum gravity theories is the idea that the fundamental structure of spacetime is discrete, and the continuum that we observe is emergent from this underlying structure. Causal sets originated as a suggested space of histories of a ``sum-over-histories'' approach to quantum theory, analogous to Feynman's path integral formulation. By discretising spacetime it also provides us with a regularization scheme to deal with UV divergences in QFT. Causal set theory models spacetime as a discrete structure of points, which are linked by a causal relation which respects the causal ordering of continuum spacetimes. Further, the macroscopic volume of a region of spacetime is proportional to the number of elements in the causal set contained in the region. Here we present an overview of the causal set theory.

The mathematical structure of causal sets is that of a partially ordered set \cite{Henson09,sorkin1990spacetime}. Thus, the standard continuum structure of spacetime is replaced by $ (\Ccal,\preceq) $, a discrete set of points $\Ccal$ -- with each point representing a spacetime event --  with a relation $\preceq$ satisfying the axioms of:
\begin{align}
x\preceq y\preceq z\implies x\preceq z&, \qquad \textrm{\it transitivity} \label{axioms1}\\
x\preceq y \text{ and } y\preceq x \implies x=y&, \qquad \textrm{\it acyclicity }\label{axioms2}\\
|I(x,y)|<\infty&,\qquad\textrm{\it local finiteness}\label{axioms3}
\end{align}
where 
\be\label{df:I}
 I(x,y)=\{z\in \Ccal\mid x\preceq z\preceq y\}
\ee
  is the set known as the causal interval (or Alexandrov set). We write $x\prec y$ if $x\preceq y$ and $x\neq y$. The physical interpretation of $x\preceq y$ is that the event $x$ is in the causal past of $y$ (allowing for equality). Some of the main building blocks of the theory are defined as follows:
\begin{df}\label{df:link} 
	A \textit{chain} in a causal set $(\Ccal,\preceq)$ is a totally ordered subset of $\Ccal$.
	A pair $x,y\in \Ccal$ is a \textit{link}, denoted $x\link y$, if  $x\prec y$ and there is no $w\in \Ccal$ such that  $x\prec w\prec y$. In particular, if $x\link y$, then $I(x,y)=\{x,y\}$. 
	A \textit{path}	is a chain such that each pair of consecutive elements is a link. 
\end{df}
Thus, a finite chain of length $n$ is an ordered set of elements
	\be
	x_1\prec x_2\prec\dots\prec x_{n-1}\prec x_{n}\,,
	\ee
while a finite path of length $n$ is an ordered set of elements with
\be
x_1\link x_2\link\dots\link x_{n-1}\link x_{n}\,.
\ee
We will denote such path by $(x_1,\dots, x_n)$. This is an analogue of a causal curve.

In analogy with the continuum, it is convenient to introduce the 
following notation.
\begin{df}
Given $x\in \Ccal$, we introduce the causal past 
\[
J^{-}(x)=\{y\in\Ccal|y\preceq x\}
\]
of $x$; for a subset $A\subset \Ccal$ we write $J^-(A)=\cup_{x\in A} J^-(x)$. 
It is also useful to define $J^-_0(x)=J^-(x)\setminus\{x\}$ and $J^-_0(A)=\cup_{x\in A} J_0^-(x)$.
Analogously, we also introduce the causal future $J^{+}$.
\end{df}

An interesting class of causal sets are those that can be formed by taking a subset of points in a Lorentzian manifold $\Mcal=(M,g)$, with a (subset of) the inherited causal order. These are called \emph{embedded causal sets}. For example a regular diamond lattice can be embedded within Minkowski spacetime.
 
The discussion of continuum limits can be facilitated by considering causal sets equipped with a length scale, forming triples $(\Ccal,\preceq,\ell)$. If $\Mcal=(M,g)$ is a time-oriented $D$-dimensional Lorentzian manifold, a sequence
$(\Ccal_n,\preceq_n,\ell_n)$ ($n\in\NN$) of embedded causal sets will be said to have $\Mcal$ as its continuum limit if, for all $n$,
\begin{equation}
    \Ccal_n\subset\Ccal_{n+1},\qquad
    p\preceq_n q\implies p\preceq_{n+1}q,
\end{equation}
$\Ccal\doteq\bigcup_n \Ccal_n$ is dense in $M$ and, for all $p,q\in\Ccal$, 
\begin{equation}
    \lim_{n\to\infty} \ell_n^D |I_{\Ccal_n}(p,q)| =
    \text{Vol}_{\Mcal}(J_{\Mcal}^+(p)\cap J_{\Mcal}^-(q)).
\end{equation} 
We emphasize that these continuum limits are to be regarded as theoretical constructions: a universe that actually is a causal set would be fundamentally discrete with a continuum as an approximation at suitable scales. Our continuum limits provide one way to control such approximations.

In the causal set literature, one often considers randomly chosen locally finite embedded causal subsets of a given $D$-dimensional Lorentzian manifold $\Mcal=(M,g)$. There is a specific choice of a measure -- the Poisson measure -- on these subsets, so that, fixing a length scale $\ell$, the probability that a randomly chosen $\Ccal$ has $n$ points in a volume $V$ is:
\begin{equation}\label{PoissonProc}
\text{Prob}(|\Ccal\cap V|=n)=\frac{(\rho V)^ne^{-\rho V}}{n!},
\end{equation}
where $\rho=\ell^{-D}$ is the fundamental density. In particular, the expected number of points in a given spacetime volume $V$ obeys 
\begin{equation} 
\ell^D \mathbb{E} |\Ccal\cap V|=   \text{Vol}_\Mcal(V).
\end{equation}
Causal sets obtained in this way are called \textit{sprinklings}. To generate the link matrix in a sprinkling, we say that two elements $p$ and $q$ are linked if and only if their Alexandrov neighbourhood  does not contain another element of the sprinkling (see \cite{Henson09} for details).
\begin{rem}
It is important to note that a generic embedded causal set  $\chi:\Ccal\hookrightarrow \Mcal$, does not inherit the local structure of that spacetime, since there could be direct links between points $x,y\in \Ccal$ such 
 that  $\chi(x),\chi(y)\in M$ appear widely separated with respect to the metric $g$. 
\end{rem}

For concrete computations, we typically label the elements of a causal set by natural numbers. One may always 
choose a \emph{natural labelling} which respects the ordering such that if $x_n\prec x_m $ then $n<m, n,m\in\mathbb{N}$ \cite{Sorkin11}. \emph{We caution the reader that, when we represent a field on a causal set by a column vector $\phi_n$ of its values at $x_n$, the elements at the top of the vector correspond to the values of $\phi$ in the far past.} 
Once this labelling has been found, two adjacency matrices can be constructed, both of which are lower triangular matrices which vanish on the diagonal:
\begin{df}
	The \textit{causal} or \textit{chain} matrix contains all of the relations between any causally related spacetime elements:
	\begin{equation}
	C_{xy}=\begin{cases}
	1, & \text{if } y\prec x \\
	0, & \text{otherwise.}
	\end{cases} 	\label{C-matrix}
	\end{equation} 
	The \textit{link} matrix is given by
	\begin{equation}
	L_{xy}=\begin{cases}
	1, & \text{if } y\link x \\
	0, & \text{otherwise}
	\end{cases} 	\label{L-Matrix}
	\end{equation}
	where $\link$ was introduced in Definition~\ref{df:link}. 
\end{df}

\subsection{Causal set Cauchy surfaces}
Here we consider natural analogues to the notion of a Cauchy surface for causal sets. We start with \textit{maximal anti-chains} \cite{MRS06}.
\begin{df}
	An	\textit{anti-chain}	is a collection of elements $\Sigma\subset\Ccal$ such that $\forall x,y\in \Sigma$ neither $x\prec y$ nor $y \prec x$. A \textit{maximal anti-chain} is an anti-chain such that any element not
	in it is related to it, which partitions the causal set as a union of mutually disjoint subsets $\Ccal = \Sigma \cup J^+_0(\Sigma)\cup J^-_0(\Sigma)$.
\end{df}
A maximal anti-chain can be regarded as a generalisation of an instantaneous time hypersurface. 
For our purposes it will be more convenient to generalise the idea that a Cauchy surface is a set on which initial data can be posed for normally hyperbolic operators. For second order operators in the continuum, the initial data consists of the field and its normal derivative; in the discrete setting the derivative is replaced by a finite difference and it is therefore convenient to replace maximal anti-chains by thickened objects that we will call \textit{Cauchy slices}. 

We start by defining Cauchy slices identified as  \textit{future}/\textit{past infinity}. In a finite causal set -- our main interest -- one can always find elements that have no future or no past. 
The definition of past and future infinity is formulated in terms of \textit{layers}, as introduced in \cite{Sorkin09}, which give a meaning to the spacetime separation of two points by using the notion of the causal interval \eqref{df:I} to find a `proximity measure' $n$ between two points: 
\begin{equation}
n(x,y)=|I(y,x)|-1\,. 
\end{equation}
Using this, the $i$'th layer below $x\in\Ccal$, $L_i(x)$, can be defined as:
\begin{equation}
L^-_i(x)\doteq\{y\in\Ccal\mid y\prec x,~n(x,y)=i\}. \label{layers}
\end{equation}
One can also define dual layers using the reversed order: 
\begin{equation}
L^+_i(x)\doteq\{y\in\Ccal\mid y\succ x,~ n(y,x)=i\},
\end{equation}
Figure \ref{fig:Layers} illustrates how the layers are defined for a regular lattice and a simple sprinkling. 

\begin{figure*} 
		\includegraphics{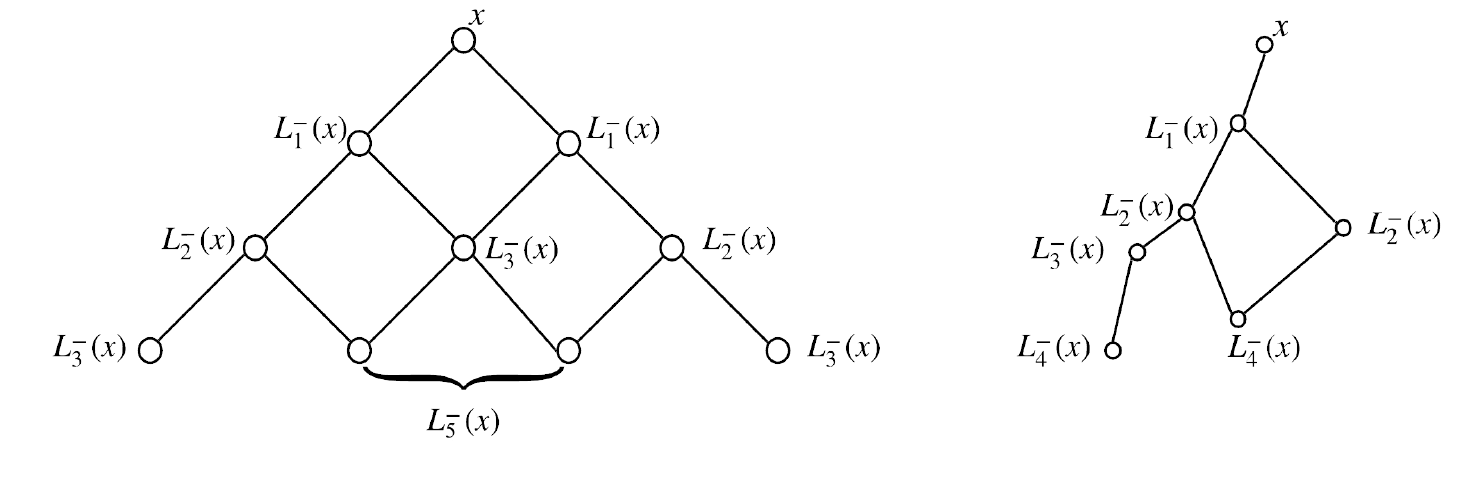} 
	\caption[Illustration of how layers are defined]{\label{fig:Layers} An illustration of how layers are defined on a regular diamond lattice (left) and a less symmetric causal set (right).}
\end{figure*}

The notion of future and past infinity is formalised as follows.
\begin{df}[Past and Future Infinity]
For $n\in\NN$, the $n$-layer \textit{past infinity} $ C_{n}^{-} $ is defined by
\begin{equation}\label{eq:pastinf}
C_{n}^-\doteq\{x\in \Ccal\mid L^-_i(x)=\emptyset, \forall i\geq n \} =\{x\in\Ccal\mid n(x,y)<n,~\forall y\prec x 
\}.
\end{equation}
Similarly, the $n$-layer \textit{future infinity} is defined by
\begin{equation}
C_{n}^+ \doteq\{x\in \Ccal\mid L^+_i(x)=\emptyset, \forall i\geq n \} 
=\{x\in\Ccal\mid n(y,x)<n,~\forall y\succ x 
\}.
\end{equation}
\end{df}
As $n(x,y)\ge 1$ for $y\prec x$, one notes that $C_1^-$ consists of all points having no predecessor; similarly, $C_1^+$ consists of those with no successor. Further, if $y\prec x\in C^-_n$ then all $z\prec y$ obey $n(y,z)<n(x,z)<n$ so also $y\in C_n^-$, i.e., $C_n^-$ is closed under taking predecessors, and $C_n^+$ is closed under taking successors. It follows that these sets are causally convex: that is, if $p,q\in C_n^\pm$ then $I(p,q)\subset C_n^\pm$. 

It will be convenient to represent the past and future infinities by diagonal matrices:
\begin{equation} 
(S_n^{\pm})_{xx}=\begin{cases}
1, & \text{if } x\in C_n^\pm \\
0, & \text{otherwise}\,.
\end{cases} \label{p/finfunit}
\end{equation}

We define a Cauchy slice in a general causal set as follows: take any maximal anti-chain $\Sigma$ and consider either the $n$-layer past infinity region within $J^+(\Sigma)$ or the $n$-layer future infinity region within $J^-(\Sigma)$. Note that our Cauchy slices do not in general correspond to the `thickened anti-chains' defined in~\cite{MRS06} (we thank an anonymous referee for an instructive counterexample).

For dynamics governed by second order differential equations, we expect to need at least two layers in a Cauchy slice to adequately specify the initial data. Depending on the discretization of the d'Alembertian employed, it may be necessary to include more layers. This is the case, for example, for the discretized d'Alembertian proposed in \cite{Sorkin09,DG13,Glaser_2014} (our Eq.~\eqref{Sorkin's Dalem}) and discussed further below in section~\ref{sec:waveeq}.

Finally, another proximity measure between two points is provided  by the notion of a \textit{rank}.
\begin{df}
	Given $x\in \Ccal$, the rank of  $y\in \Ccal$ relative to $x$, $\rk(x,y)$, is defined as the minimal number of links in a path from $y$ to $x$ (i.e., one less than the minimal length of such a path). The rank is infinite if there is no path from $y$ to $x$ and $\rk(x,x)=0$.
\end{df}
Its relationship to the past and future infinity sets is expressed by the following lemma.
\begin{lemma}\label{lem:RankLayer}
	Define $R^-_n$ to be the space of points that have no points to their past of rank $n$ or higher, i.e. 
	\be
	R_n^-\doteq \{x\in\Ccal|\rk(x,y)<n\,,\ \forall y \prec x\}\,.
	\ee
	Then we have 
	\be
	C_n^-\subset R_n^- 
	\ee
	and for the special case $n=2$, we have $C_2^-= R_2^-$. Analogous results hold for $R_n^+$ and $C_n^+$, where past is replaced with future.
\end{lemma}
\begin{proof}
	Firstly, note that if $y\prec x$ and $\rk(x,y)\ge n$ then the cardinality of the Alexandrov set $I(y,x)$ is at least $n+1$, so $y\in L_i^-(x)$ with $i\geq n$, i.e., $x\notin C^-_n$. Now, turning this argument around, if $L_i^-(x)=\varnothing$ for all $i\geq n$, then for all $y\prec x$ we must have $\rk(x,y)<n$, so $x\in R_n^-$. 
	
	In the special case $n=2$, $\rk(x,y)<2$ for all $y\prec x$ implies that in fact $y\prec * x$, so $I(y,x)=2$ and hence $y\in L_1^-(x)$. As this holds for all $y\prec x$, we conclude that $x\in C^-_2$.
\end{proof}

\section{Classical field theory on a fixed causal set}
In this work, we take the algebraic viewpoint and introduce the  classical theory on a fixed causal set by constructing an appropriate Poisson algebra. We focus on the example of the real scalar field, starting with a discussion of the relevant kinematical structures and then discussing discretized d'Alembertian operators and their Green functions in some detail. From there we move to a discussion of a Peierls bracket and then to construct algebras describing free and interacting field theories.

\subsection{Kinematical structure}\label{sec:kinematics}
 
\begin{df}[Real scalar field on a causal set]
	The real scalar field on a causal set $\Ccal$ of size $N$ has a configuration space $\Ecal(\Ccal)$ consisting of maps $\phi:\Ccal\rightarrow\RR$,	with a vector space structure of pointwise operations.
	Given a natural labelling of $\Ccal$ by $\{1,\dots,N\}$, we identify
	$\Ecal(\Ccal)\cong \RR^N$, regarded as a space of column vectors. We use the notation $\phi_i$, $i=1,\dots N$ for the components of field $\phi$, remembering that low values of the index correspond to spacetime events in the `far past'. 
\end{df} 
If the causal set is equipped with a length scale, it becomes possible to discuss dimensionful fields, saying that $\phi$ has dimension $d$ to mean dimensions of $[\text{length}]^d$. On the basis that a length is a quantity whose numerical value increases in inverse proportion to a decrease in the units of length, a scalar field of dimension $d$ on $\Ccal$ should transform under a change of length scale $\ell\mapsto {\lambda}\ell$ by 
\begin{equation}
\phi(p)\mapsto {\lambda}^{-d} \phi(p).
\end{equation}
More generally, if $(\Ccal,\preceq,\ell)$ is embedded within
$(\Ccal',\preceq',\ell')$, a dimension $d$ scalar field $\phi'$ on $\Ccal'$ pulls back to a field on $\Ccal$ defined by
\begin{equation}
\phi(p) = (\ell'/\ell)^d \phi'(p).
\end{equation}

We wish to consider a field theory on a causal set that has a claim to be an analogue of the wave equation on a continuum Lorentzian spacetime. It is therefore necessary to be able to compare the continuum and discrete situations. One way of doing this is through a suitable continuum limit. 
First, a dimension $d$ scalar field $\phi$ on a time-oriented Lorentzian spacetime $\Mcal$ may be pulled back to a function $\phi_\Ccal(p)= \ell^{-d}\phi(p)$ on any causal set
$(\Ccal,\preceq,\ell)$ embedded in $\Mcal$.
This viewpoint allows us to work solely with numerical scalar fields on causal sets. 

Next, consider a situation in which $\Mcal$ is
the continuum limit of a sequence of causal sets $(\Ccal_n,\preceq_n,\ell_n)$, as defined in section \ref{sec:CausalSets}. We say that a sequence of functions 
$\phi_{\Ccal_n}:\Ccal_n\to\RR$ has a continuum limit as a continuous dimension $d$ field $\phi:\Mcal\to\RR$ if
\begin{equation}
   \phi(p)= \lim_{n\to\infty} \ell_n^d \phi_{\Ccal_n}(p)  
\end{equation}
for all $p\in \Ccal\doteq\bigcup_n \Ccal_n$. 

For example, we will shortly discuss discretised analogues of the d'Alembertian, which changes dimensions by two powers of length in the continuum. If the continuum limit $\phi$ just described is twice continuously differentiable, a family of discrete d'Alembert operators $P_n$ on a sequence of causal sets $(\Ccal_n,\preceq_n,\ell_n)$ would therefore be expected to obey
\begin{equation}\label{eq:contcrit}
\ell_n^{d-2} (P_n \phi_{\Ccal_n})(p) \longrightarrow (\Box\phi)(p)
\end{equation}
as $n\to\infty$ for every $p\in \bigcup_n \Ccal_n$.  
 
Observables are defined similarly to the continuum case:
\begin{df}[Causal set observables]\label{moreobs}
	Observables on a causal set $\Ccal$ are smooth maps from the configuration space $ \Ecal(\Ccal) $, to $\CC$, i.e. they are elements of $\Ci(\Ecal(\Ccal),\CC)\equiv \Fcal(\Ccal)$. The space of observables is equipped with the natural structures of addition
	\[
	(F+G)(\ph)=F(\ph)+G(\ph)
	\] 
	and multiplication
	\[
	FG(\ph)=F(\ph)G(\ph)\,.
	\]
\end{df}	
A special case is given by linear observables, defined, for each $f\in\CC^N$ ($|\Ccal|=N$), by
	\begin{equation}\label{causet obs}
	\Phi_f:\RR^N\to\CC;\ \Phi_f(\phi)\doteq f^i\phi_i\equiv f^T\phi,
	\end{equation}
	where $\phi\in\Ecal(\Ccal)\cong\RR^N$ and we have used the Einstein summation convention for repeating indices.
	 The space of linear observables is denoted by $ \Xcal(\Ccal) $. This space has a vector space structure, inherited from $\Fcal(\Ccal)$, and this structure is compatible with the addition on the labeling space, i.e.
	  $ \Phi_g+\lambda \Phi_h=\Phi_{g+{\lambda} h} $, for any $ \lambda\in\CC $. 
	\begin{rem}\label{rem:prod}
	In the last expression of formula \eqref{causet obs}, we made implicit use of the Euclidean metric on $\RR^N$ and the induced inner product. This metric allows us to identify elements of $\RR^N$ with observables and will be used to raise and lower indices. To see how this is consistent with the viewpoint on continuum limits and dimensions taken earlier, consider $\phi,f$ that are smooth functions on $D$-dimensional $\Mcal$ that have supports intersecting compactly (for simplicity) and
	 have dimensions $d_\phi,d_f$. Then one has
	\be
	\Phi_{f_\Ccal}(\phi_\Ccal)= \ell^{-d_\phi-d_f} \sum_{p\in\Ccal} \phi(p)f(p)\,,
	\ee
	so if we have  a sequence of functions 
	$\phi_{\Ccal_n}:\Ccal_n\to\RR$ and 	$f_{\Ccal_n}:\Ccal_n\to\CC$ with continuum limits $\phi$ and $f$ respectively, then
\begin{align}
\lim_{n\to\infty}		\Phi_{f_{\Ccal_n}}(\phi_{\Ccal_n}) &= \ell^{-d_\phi-d_f-D}\lim_{n\to\infty}	 \ell^D \sum_{p\in\Ccal_n} \phi_{\Ccal_n}(p)f_{\Ccal_n}(p) \notag \\
& = \ell^{-d_\phi-d_f-D}
	\int_M f(x) \phi(x) d\mu_g(x).
\end{align}
	Hence the continuum analog of the inner product is the choice of a volume form $\varepsilon$ on spacetime $\Mcal=(M,g)$ (take e.g. the invariant volume form $d\mu_g$ induced by the Lorentzian metric), allowing one
	 to identify $f\in\Ci_c(M,\CC)$ with the observable
	\be\label{eq:smeared:cont}
	\Phi_f(\ph)=\int_M f(x) \ph(x) \varepsilon(x)\,.
	\ee	
	The consequences of choosing a different inner product will be discussed in more detail in section~\ref{sec:free}.
	\end{rem}

	Next we introduce the notation for functional derivatives. The functional derivative of $F\in\Fcal(\Ccal)$ at point $\ph\in\Ecal(\Ccal)$ in the direction of $\psi\in\Ecal(\Ccal)$ is defined by:
	 \be
	 \left<F^{(1)}(\ph),\psi\right>\doteq \lim_{t\rightarrow 0}\frac{1}{t}(F(\ph+t\psi)-F(\ph))\,,
 	 \ee 
 	 where $t\in\RR$. We will also use the notation 
 	 \be
 	 F^{(1)}(\ph)\equiv \frac{\delta F}{\delta \phi}(\ph)\,.
 	 \ee
	 Note that since $\Ecal(\Ccal)\cong \RR^N$, the functional derivative  $\frac{\delta F}{\delta \phi}(\ph)$ at point $\ph$ is a linear $\CC$-valued functional on $\RR^N$  and therefore can be identified
	 with an element of $\CC^N$ and we write its components as $\frac{\delta F}{\delta \phi_i}(\ph)$, $i=1,\dots,N$.
	 
	 We introduce a product on $ \Xcal(\Ccal) $, induced by the component-wise multiplication of the smearing functions $g\in\CC^N $, or the \textit{Hadamard product}: 
	 \begin{equation}\label{eq:Hadamardproduct}
	 \Phi_g*\Phi_h\doteq \Phi_{g*h}\,,
	 \end{equation} 	 
	 where $(g*h)^i=g^ih^i$, with no summation over the repeated indices. 
	 
	 Another natural product on $ \Xcal(\Ccal) $ is the pointwise product of observables, inherited from $\Fcal(\Ccal)$:
	 $$ (\Phi_g\cdot \Phi_h)(\phi)=\Phi_g(\phi)\Phi_h(\phi)\,,$$
	 which does not leave $ \Xcal(\Ccal) $ invariant. Let $\Fcal_{\reg}(\Ccal)$ denote the subalgebra of $\Fcal(\Ccal)$ generated by $ \Xcal(\Ccal) $ with respect to $\cdot$. This is the analog of \textit{regular functionals} in continuum pAQFT. They form a $*$-algebra, where the $*$ operation is just the complex conjugation.

\subsection{Classical dynamics}\label{sec:classdyn}
\subsubsection{Discretized retarded wave equations}\label{sec:waveeq}

As in continuum QFT, we will construct the interacting theory as a perturbation of a free field equation. The starting-point is therefore a suitable discretization of the continuum field equation
\begin{equation}
\Box\phi = f
\end{equation}
to a causal set. Several possible causal set d'Alembertians or `box operators' have been discussed previously \cite{Sorkin09,DG13,Glaser_2014,Aslan_etal:2014}, and we will give a specific example below as well as introducing a new type of box operator.  We study equations taking the form
\begin{equation}\label{eq:WE1}
P \phi = Kf 
\end{equation}
neglecting edge effects for the moment -- they will be discussed in Sec.~\ref{sec:Cauchy}. Here $f,\phi\in \Ecal(\Ccal)$ are the source and solution respectively, while $P$ and $K$ 
are linear maps on $\Ecal(\Ccal)$. The map $K$ is newly introduced here, and can absorb factors (it sometimes turns out to be more convenient to discretise $\frac{1}{2}\Box$ rather than $\Box$) but also provides additional freedom to determine the way in which a continuum source is discretized.

Various requirements on $P$ were set down in~\cite{Aslan_etal:2014}. First, in addition to linearity, $P$ is required to be a \emph{retarded operator}, meaning that $(P\phi)_p$ is a linear combination of $\phi_q$ with $q\preceq p$. We also require that $K$ be retarded in this sense and that both operators are real.
As will be seen, this requirement ensures the causal nature of solutions to~\eqref{eq:WE1}. Second, the prescription for constructing $P$ and $K$ should be independent of the way in which the causal set is labelled -- a covariance requirement. In~\cite{Aslan_etal:2014} a requirement of `neighbourly democracy' is imposed, namely that all points in the same layer below $p$ contribute with equal weight to $(P\phi)_p$; we will not impose this and indeed will introduce an `undemocratic' example that may be defined on causal sets with a \emph{preferred past structure}. Our last general requirement is that each $(P\phi)_p$ should have nontrivial dependence on $\phi_p$; in~\cite{Aslan_etal:2014} it was assumed that the coefficient should be independent of $p$, but one could certainly envisage prescriptions in which the coefficient was variable and determined by the statistics of the causal order, restricted to the past of $p$.

In a natural labelling of the causal set, these requirements ensure that $P$ is in particular lower triangular and its diagonal entries are all nonvanishing. Consequently, $P$ is invertible and it may easily be seen that $P^{-1}$ is also a retarded operator. Clearly the solution to~\eqref{eq:WE1} is then
$\phi = E^+ f$, where
\begin{equation}
E^+ \doteq P^{-1}K
\end{equation}
defines the \emph{retarded Green operator}. Note that the composite of retarded operators is retarded. As in~\cite{Sorkin}, we define the \emph{advanced Green operator} to be
\begin{equation}
E^-\doteq (E^+)^T,
\end{equation}
and the advanced-minus-retarded\footnote{This differs from the convention used e.g. in \cite{DF02,Book,HR}, where the operator $P$ in the continuum is $-\Box$, rather than $\Box$, so that $E^\pm$ in those references are Green functions for $-\Box$ and $E$ ends up with the opposite sign.} operator is the anti-symmetric matrix
\begin{equation}\label{df:E}
E =E^- - E^+ =  (E^+)^T-E^+.
\end{equation}
By construction, $(E^-f)_p$ is a linear combination (with real coefficients) of $f_q$ with $p\preceq q$, and therefore an \emph{advanced operator} by analogy with previous definitions. We have followed the existing literature by emphasising the retarded equations and Green operators as the starting-point. It would be possible, though less physically well-motivated, to base the discussion on advanced operators. 

As a specific example, we recall the d'Alembertian defined in~\cite{Sorkin09} (we multiply by a factor of $\tfrac{1}{2}$ and adapt to our sign conventions)
\begin{equation}
(P_S\phi)_p\doteq \phi_p-2\left(\sum_{q\in L^-_1(p)}\phi_q-2\sum_{q\in L^-_2(p)}\phi_q+\sum_{q\in L^-_3(p)}\phi_q\right) \,. \label{Sorkin's Dalem}
\end{equation}
Sorkin also included a factor of $\ell^{-2}$, where $\ell$ is the fundamental length scale associated with the sprinkling, which is not present here because of the way we treat dimensionful fields. 
In matrix form, 
\begin{equation}
(P_S)_{pq}=\begin{cases}
1,&p=q\\
-2,4,-2,\ p\neq q,&n(p,q)=1,2,3\ \textrm{respectively}\\
0,&\textrm{otherwise,}
\end{cases}
\end{equation}
and is lower-triangular in a natural labelling.
In \cite{Sorkin09} the continuum limit of the operator \eqref{Sorkin's Dalem}, averaged over sprinklings into two-dimensional Minkowski space $\M_2$, was shown to be the continuum d'Alembertian $\frac{1}{2} \square$. Generalizations exist to $d$-dimensional spacetimes for $d>2$, but involve more layers and different coefficients, to obtain the correct continuum limit for sprinklings into $\M_d$ \cite{BDD10,DG13,Glaser_2014}.  

\subsubsection{Causal sets with a preferred past structure}\label{sec:prefpast}

As an alternative to the principle of neighbourly democracy, we propose a new type of discretized d'Alembertian for causal sets, which will be investigated in more detail elsewhere. It is based on
a `preferred past structure' defined as follows. 
\begin{df}\label{df:pref:future}
	Given a causal set $\Ccal$, a \textit{preferred ($2$-step) past structure} is a map ${\Lambda}:\Ccal\setminus {C_2^-}\to\Ccal$  so that, for each $p\in \Ccal \setminus {C_2^-}$, the \textit{preferred past} ${\Lambda}(p)$ of $p$ is a point of rank $2$ in the past of $p$. The corresponding \emph{preferred past matrix} is a lower triangular matrix with vanishing diagonal entries, given by
	\begin{equation}
	\Lambda_{xy}=  \delta_{{\Lambda}(x)\,y}= \begin{cases}
	1 & \text{if $y={\Lambda}(x)$,} \\
	0 & \text{otherwise.}
	\end{cases}
	\end{equation}
\end{df}

We will regard the causal interval between ${\Lambda}(p)$ and $p$ as an elementary non-atomic volume in the causal set. Lemma~\ref{lem:RankLayer} shows that every point outside $C_2^-$ has points of rank $2$ in its past. Therefore every causal set in which every point has at most finitely many past-directed links (and therefore at most finitely many points of rank $2$ in its past) admits (at least one) preferred $2$-step past structure.

In general, there may exist more than one possible preferred past structure, in which case a choice must be made. Ideally, there should be some additional rule for selecting $\Lambda$ in a given causal set to restrict the choice. For example, one could require that ${\Lambda}(p)$ of $p$ is a point with maximal layer number (among all the points in the past of $p$ of rank $2$). Consider the regular diamond lattice in $\M_2$, a portion of which is illustrated in the left-hand part of Fig.~\ref{fig:Layers}. The points of rank $2$ below $x$ are in the third row, and there is a unique point with maximal layer number, i.e., the centre point in that row, belonging to the third layer below $x$; therefore the `maximal layer rule' selects a unique preferred past structure in this example. For general sprinklings, it is not yet clear to us what rule is the most appropriate one. Other possible rules for selecting $\Lambda$ will be investigated in our future work.

Using a preferred past structure, we may introduce a new type of 
discretised retarded d'Alembert{\-}ian. An example, developed especially with two-dimensional continuum spacetimes in mind, is given as follows: 
\begin{equation}\label{NewdAl}
(P_{\Lambda}\phi)_p= \begin{cases}\phi_p - 2\displaystyle\left(\mean_{{\Lambda}(p)\prec q \prec p}\phi_q\right) + \phi_{{\Lambda}(p)} & p\notin C_2^- \\ \phi_p & p\in C_2^-,  
\end{cases}
\end{equation}
where 
\begin{equation}
\mean_{q\in U}\phi_q = |U|^{-1}\sum_{q\in U} \phi_q
\end{equation} 
is the arithmetic mean taken over a subset $U\subset\Ccal$. Here it is necessary to treat points in $C_2^-$ separately because they do not have preferred pasts. 
Note that $(P_\Lambda\phi)_p$ involves a sum over points of at most rank $2$ below $p$ --- to be precise, those in the causal interval between $p$ and its preferred past $\Lambda(p)$ --- and that the coefficients associated with each contributing point are determined by the rank relative to $p$ and so are independent of the way that the causal set is labelled. 

It is convenient to present $P_\Lambda$ as a matrix. To this end, we define a lower triangular matrix $\Omega$ with vanishing diagonal given by 
\begin{equation}
\Omega_{pq}= \begin{cases}
1 &  \Lambda(p)\prec q\prec p\\ 
0 & \text{otherwise,}
\end{cases}
\end{equation}
which encodes information about the causal intervals associated with the preferred past structure, and also a diagonal weight matrix $W$,  
\begin{equation}
W_{pp} = \begin{cases} (\sum_{q}\Omega_{pq})^{-1} & p\notin C_2^- \\
0 & p\in C_2^-\,.
\end{cases}
\end{equation}
The second case deals with edge effects to avoid an infinite value. In fact its value will not matter. Then the discretised operator may be written as
\begin{equation}
P_\Lambda = \1 + \Lambda -2 W \Omega.
\end{equation} 

One reason for regarding $P_{\Lambda}$ as a causal set analogue of half the d'Alembertian is that it produces a valid discretisation of the continuum operator $\frac{1}{2}\Box$ using regular diamond lattices. Consider 
the lattice $\{(m\ell\sqrt{2},n\ell\sqrt{2}):m,n\in\ZZ\}$ embedded in $\M_2$, 
using $(u,v)$-coordinates related to the standard inertial Minkowski coordinates by $u=t-x$, $v=t+x$. The continuum metric and d'Alembertian are
$ds^2=du\,dv$ and $\Box = 4\partial_u\partial_v$.
Each lattice cell therefore has spacetime volume $\ell^2$ (explaining the factor of $\sqrt{2}$ above), so $\ell$ is a natural length scale associated with the lattice and indeed one has
\begin{equation}
	\ell^2 |I(p,q)| \sim \textrm{Vol}_{\M_2}(J_{\M_2}^+(p) \cap J_{\M_2}^-(q))\,,
\end{equation}
when $p$ and $q$ are widely separated lattice points and $J_{\M_2}^\pm$ on the right-hand side refer to the causal future/past of the continuum spacetime.
The sequence of such lattices with $\ell_r=\ell/r$ ($r\in \NN$) has $\M_2$ as its continuum limit, associating the length scale $\ell_r$ with each. We denote the corresponding causal sets by $(\Ccal_r,\preceq,\ell_r)$, with ordering $p\preceq q$ in all cases determined by the causal order of $\M_2$. 

\begin{figure} 
	\includegraphics{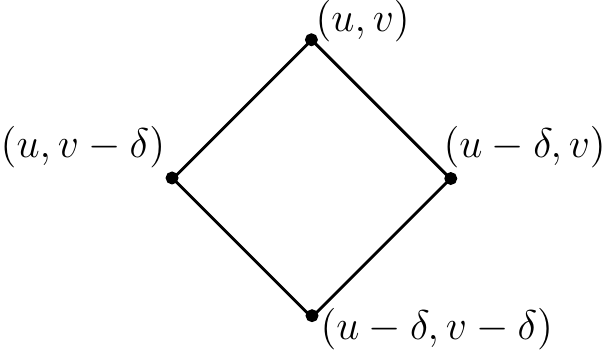}  
	\caption{\label{fig:DiamMin}Parametrization of points in a segment of a regular diamond lattice, embedded into $\M_2$, where $\delta=\ell\sqrt{2}$.}
\end{figure}
\begin{widetext}
	Suppose, for simplicity, that $\phi$ is a smooth dimensionless scalar field on $\M_2$, which pulls back to causal set $\Ccal_r$ by restriction. Writing $\delta_r=\ell\sqrt{2}/r$, we have
\begin{equation}
({P_{\Ccal_r, \Lambda}}\phi_{\Ccal_r})(u,v)=\phi(u,v)-\phi(u-\delta_r,v)-\phi(u,v-\delta_r)+\phi(u-\delta_r,v-\delta_r)\,. \label{eq:FieldParam}
\end{equation}
Taking Taylor series to second order, 
\begin{align}
\phi(u-\delta_r,v) &= \phi(u,v) -\delta_r \partial_u\phi(u,v) + \frac{\delta_r^2}{2} \partial_u^2\phi(u,v) + O(\delta_r^3)\notag\\
\phi(u,v-\delta_r) &= \phi(u,v) -\delta_r \partial_v\phi(u,v) + \frac{\delta_r^2}{2} \partial_v^2\phi(u,v) + O(\delta_r^3)\notag\\
\phi(u-\delta_r,v-\delta_r) &= \phi(u,v) -\delta_r( \partial_u\phi(u,v)+\partial_v\phi(u,v)) \notag\\ 
&\qquad + \frac{\delta_r^2}{2} (\partial_u^2\phi(u,v)+2\partial_u\partial_v\phi(u,v)+ \partial_v^2\phi(u,v)) + O(\delta_r^3)
\end{align}
with error terms uniform in $r$. Therefore, 
\begin{equation}
({P_{\Ccal_r, \Lambda}})(u,v)=\delta_r^2\partial_u\partial_v\phi(u,v)+O(\delta_r^3)\,, \label{eq:FieldParam2}
\end{equation}
and it follows that
\begin{equation}
\ell_r^{-2}({P_{\Ccal_r, \Lambda}}\phi_{\Ccal_r})(u,v)\longrightarrow  2\partial_u\partial_v\phi(u,v)= \tfrac{1}{2}(\Box\phi)(u,v)
\end{equation}
as $r\to\infty$, which is the claimed continuum limit. 
\end{widetext}
Given this result, a natural choice for $K$ is to set $K_\Lambda=\tfrac{1}{2}\1$. However, this prescription is not the only possibility and should be reconsidered near the past boundary of the causal set if there is one. See further comments below. We remark that Sorkin's operator $P_S$ does not have the continuum limit $\tfrac{1}{2}\Box$ on the regular lattice; instead, it is adapted to sprinklings into $\M_2$.

The extent to which $P_\Lambda$, or similar operators, approximate the d'Alembertian in higher dimensions on regular or sprinkled lattices will be reported elsewhere. Our main purpose in introducing it here is to illustrate the point that there are discretised d'Alembertians that do not obey the neighbourly democracy principle, but are still naturally associated with the causal set, augmented by a preferred past structure. Our hope is that some generalization of this ansatz could be applied in arbitrary dimensions in such a way that the dimension itself is not an input (as in the proposals \cite{BDD10,DG13,Glaser_2014}), but an emergent quantity. Typically, there will be more than one possible preferred past associated with a given causal set. One could remove the element of choice by averaging $P_\Lambda$ over all such possibilities.

\subsubsection{Retarded Green function for $P_\Lambda$ on the regular diamond lattice}

The retarded Green function may be computed exactly for $P_\Lambda$ on the regular diamond lattice in $\M_2$ for various choices of operator $K$, which 
may help to illustrate the additional freedom that it represents. The starting observations are that $\Omega$ precisely coincides with the link matrix $L$, and that $W=\tfrac{1}{2}\1$. Thus we have
\begin{equation}
P_{\Lambda}= \1 - L + \Lambda.
\end{equation} 
\begin{lemma}
	For the regular diamond lattice, and taking $K_\Lambda=\tfrac{1}{2}\1$, the retarded and advanced Green functions are  
	\begin{align}\label{eq:latticeEpm}
	E_\Lambda^+ &= \tfrac{1}{2}(P_{\Lambda})^{-1} = \tfrac{1}{2}(\1-L+{{\Lambda}} )^{-1}=\tfrac{1}{2}(\1+C) \\
	E_\Lambda^- &\doteq (E_\Lambda^+)^T =\tfrac{1}{2}(\1+C^T)\,.
	\end{align}
\end{lemma}
\begin{proof}
	Direct calculation gives
	\be
	\Lambda_{pq} + [C{\Lambda}]_{pq} = \begin{cases}
	1 & q\in I_{\M_2}^-(p)\\ 0 & \text{otherwise,}
	\end{cases}
	\ee
	because $[C{\Lambda}]_{pq}=1$ if and only if $q\prec {\Lambda}(p)$. Similarly, 
	\be
	L_{pq} + [CL]_{pq} = \begin{cases}
	2 & q\in I_{\M_2}^-(p)\\ 1 & q\in \dot{J}_{\M_2}^-(p) \setminus\{p\}  \\ 0 & \text{otherwise,}
	\end{cases}
	\ee
	and one sees immediately that 
	\begin{equation}
	(\1+ C)(\1 - L + {\Lambda}) = \underbrace{{\Lambda} + C{\Lambda} -(L+ CL)}_{-C}+\1+C =\1\,,
	\end{equation}
	so $(P_\Lambda)^{-1}=\1+C$, giving~\eqref{eq:latticeEpm}. 
	The result for $E_\Lambda^-$ is obvious.
\end{proof}
This result shows that $[E^+_\Lambda]_{pq}$ takes the value $\tfrac{1}{2}$ if 
$p\preceq q$ and zero otherwise.

Let us now consider the continuum limit of these operators as the mesh of the diamond lattice tends to zero. Suppose $f$ is a smooth compactly supported function on $\M_2$ of dimension $[L]^{-2}$ for simplicity, and pull it back to the causal set $\Ccal_r$ (as in\ Sec.~\ref{sec:prefpast}) by
$(f_{\Ccal_r})_p = \ell^{2}f(p)$. Then 
\be
(E^+_{\Ccal_r, \Lambda}f_{\Ccal_r})_{p} = \frac{\ell_r^2}{2}\sum_{q\preceq p} f(q)
\ee
On the other hand, the retarded Green function on $\M_2$ is given by 
\begin{equation}
E_{\M_2}^{+}(t,x;t',y) =\tfrac 1 2 \theta((t-t')-|x-y|) =
\tfrac{1}{2} \theta(t-t')\theta((t-t')^2-|x-y|^2)
\,,
\end{equation}
where $(t,x)\in\M_2$ is a point in 2D Minkowski spacetime (with signature $(+-)$) and $\theta$ is the Heaviside step function. Thus, for fixed $(t,x)$, $E_{\M_2}^{+}$ takes value $\tfrac{1}{2}$ for $(t',y)$ inside the closed past lightcone of $(t,x)$ and vanishes otherwise. The function $E_{\M_2}^{+}f$ is dimensionless, and given by 
\begin{equation}\label{eq:EretMink2}
(E_{\M_2}^{+}f)(p)=\tfrac{1}{2} \int_{J_{\M_2}^-(p)} f(q) d\mu_g(q)
\,.
\end{equation}
It follows that
\be
(E_{\Ccal_r,\Lambda}^{+}f_{\Ccal_r})_{p} \to (E_{\M_2}^{+}f)(p)
\ee
as $r\to\infty$, because the spacetime volume of each diamond $[u,u+\delta_r]\times [v,v+\delta_r]$ is $\ell_r^2$. This shows that our operator $E_\Lambda^+$ is a valid discretisation of the continuum retarded Green function. Evidently the same will hold for the advanced Green operator. 

A different discretisation of $E^+_{\M_2}$ was considered by \cite{DS17}, namely
\begin{equation}
[E^{+}_{DSX}]_{pq}=\tfrac 1 2 C_{pq}\,,
\end{equation}
which takes the value  $\tfrac{1}{2}$ when $p\prec q$ and vanishes otherwise. This may be reproduced 
from our operator $P_\Lambda$ by changing $K_\Lambda$ to 
\begin{equation}
K_{DSX}=P_\Lambda E^{+}_{DSX} = \tfrac 1 2 P_{\Lambda}C=\tfrac 1 2(\1 - L + {\Lambda})C =
\tfrac 1 2(L - {\Lambda}) .
\end{equation}
Yet a further possibility would be to discretise $E^+_{\M_2}$ by
\begin{equation}
E^{+}_{trap} =\tfrac{1}{8}\1+ \tfrac{1}{4}(L  + CL)\,,
\end{equation}
with components $[E^{+}_{trap}]_{pq}$ equal to $\tfrac{1}{8}$ when $p=q$, $\tfrac{1}{4}$ for $q\in \dot{J}_{\M_2}^-(p)\setminus\{p\}$, $\tfrac{1}{2}$ for $q\in I_{\M_2}^-(p)$ and vanishing otherwise.
Here $\dot{J}_{\M_2}^-(p)$ is the boundary of the causal past $J_{\M_2}^-(p)$. The discretisation $E^{+}_{trap}$, which amounts to discretising~\eqref{eq:EretMink2} using the trapezium rule in $u,v$ coordinates,
corresponds to
\begin{equation}
K_{trap}=\tfrac{1}{8}(\1+L+\Lambda)\,.
\end{equation}
These definitions have the same continuum limit, $E^{+}_\M$, but correspond to different discretisations of the inhomogeneous wave equation. Consider the isolated diamond in Figure~\ref{fig:Causet_Diam}. Then:
\begin{itemize}
	\item setting $K=\frac{1}{2}\1$ corresponds to sampling the value of the test function $f$ on the diamond by taking its value only at the future-most point $f_p$;
	\item setting $K_{DSX}=\tfrac 1 2M^+$ samples $f$ by taking $f_{q_1}+f_{q_2}-f_{\Lambda(p)}$;
	\item setting $K_{trap}=\tfrac{1}{8}(\1+L+\Lambda)$ samples $f$ by taking $\tfrac{1}{4}(f_p+f_{q_1}+f_{q_2}+f_{\Lambda(p)})$.
\end{itemize}
This illustrates a basic fact that a continuum operator may have many valid discretisations, and indicates the flexibility introduced by the operator $K$. 

\begin{figure} 
	\includegraphics{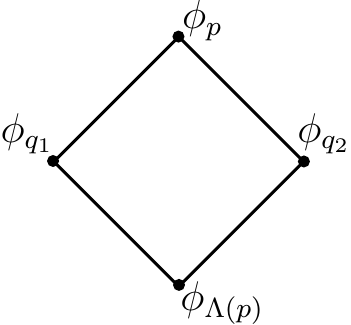} 
	\caption[A causal set diamond]{\label{fig:Causet_Diam} An isolated diamond from a regular diamond lattice.}
\end{figure}

\subsubsection{Edge effects at past infinity and the Cauchy problem}\label{sec:Cauchy}

In causal sets with a past boundary, i.e., points with no predecessors in the causal order, the form~\eqref{eq:WE1} of the wave equation given above should be reconsidered near to that boundary. In fact we have already anticipated this in our definition of $P_\Lambda$, which treats points in $C_2^-$ differently to those in the bulk. In the same way, one might expect that the operator $P_S$ might be modified for points in $C_3^-$, because these points do not have layer-$3$ predecessors.

The simplest possibility to take the edge effects into account is to first fix the discretised d'Alembertian operator $P$ outside  $C_k^-$ (the $k$-step past infinity), for some fixed $k$, using some discretization of the wave equation, and then set $(P\phi)_p=\phi_p$ for $p\in C_k^-$. This is already the case with $P_\Lambda$ for $k=2$, as defined by \eqref{NewdAl}. As for the right-hand side of the equation \eqref{eq:WE1}, the definition of $K_\Lambda$ should be modified so that the diagonal elements of $K_\Lambda$ would be $\tfrac{1}{2}$ except for entries corresponding to points in $C_2^-$, where the value would be $1$. As $P$ is then triangular with nonvanishing diagonal elements, this prescription ensures that $P$ remains invertible. The values of $\phi$ on $C_k^-$ are then treated as Cauchy data for the solution and we replace the wave equation \eqref{eq:WE1} by 
\begin{equation}\label{eq:inhom}
P\phi = K f + \phi^-\,,
\end{equation}	
where $\phi^-=S_k^-\phi$ is the projection of $\phi$ onto past infinity,  ($S_k^-$ was defined in~\eqref{p/finfunit}), and it is understood that the source $f$ should vanish in $C_k^-$. Note that $K\phi^-=\phi^-$, so one also has $P\phi=K(f+\phi^-)$. Thus our prescription allows the Cauchy data to be combined with the source in a natural way.
In these circumstances we will say that $P$ has a \emph{$k$-layer Cauchy problem}.
Recalling that Cauchy data for the scalar field in the continuum consists of both values and normal derivatives on a Cauchy surface, it is natural enough that the Cauchy data on a causal set involves values taken on at least two layers. 

Given the assumptions made on $P$ and $K$, the solution to~\eqref{eq:inhom} is
\begin{equation}
\phi = E^+ f+ E^+\phi^-.
\end{equation} 
Two special cases are of interest. First, the situation in which $\phi^-=0$,
in which case $\phi=E^+f$, in line with the continuum idea that the retarded Green function should produce solutions vanishing in the far past. Second, if $f=0$, $\phi=E^+\phi^-$ may be interpreted as the solution to the source-free equation with Cauchy data $\phi^-$, which thus takes the form
\begin{equation}\label{eq:hom}
P\phi = \phi^-
\end{equation}
and which will be called the homogeneous wave equation in the sequel.

By definition, solutions to the homogeneous wave equation are in bijection with the Cauchy data specified on $C_k^-$, which is just the space $\Ecal(C_k^-)$. Therefore the solution space is
\begin{equation}
\Ecal^+_{Sol}\doteq E^+ \Ecal(C_k^-)\,.
\end{equation}
The space of such solutions will be denoted $\Ecal^+_{Sol}$.
A particularly simple situation occurs if the solutions are also in bijection with data on future infinity, in which case
\begin{equation}
\alpha^+= S_k^+ E^+|_{\Ecal(C_k^-)}
\end{equation}
is an isomorphism $\alpha^+:\Ecal(C_k^-)\to \Ecal(C_k^+)$ that will be called the \emph{Cauchy evolution}.
Of course, this requires among other things that $C_k^\pm$ have equal cardinality.
 
As a slight digression we note that, in circumstances where the Cauchy evolution is defined, we can use it to compare the dynamics of the theory on two causal sets $\Ccal$ and $\widetilde{\Ccal}$ whose $k$-layer past and future infinity regions are in order-preserving isomorphism with each other, thus inducing linear isomorphisms 
$\iota^\pm: \Ecal(C_k^\pm)\to \Ecal(\widetilde{C}_k^\pm)$. Writing $\alpha^+$ and $\widetilde{\alpha}^+$ for the two Cauchy evolutions, the
\emph{relative Cauchy evolution} is a linear isomorphism on the solution space $\Ecal^+_{Sol}(\Ccal)$ defined by 
\begin{equation}
\mathtt{rce}(\phi)\doteq E^+ (\iota^-)^{-1}(\widetilde{\alpha}^+)^{-1} \iota^+ S_k^+\phi\,,
\end{equation}
which is an isomorphism; note that we also have the identity
\begin{equation}
S_k^+\mathtt{rce}(\phi) = \alpha^+ (\iota^-)^{-1}(\widetilde{\alpha}^+)^{-1} \iota^+ S_k^+\phi
\end{equation} 
in which the comparison of $\alpha^+$ and $\widetilde{\alpha}^+$ is apparent. Relative Cauchy evolution provides a way of discussing the response to changes in causal set geometry by reference to solutions of the wave equation on the unperturbed causal set. It was first introduced in \cite{BFV}, where it was formulated in locally covariant QFT on continuum spacetimes, for perturbations in the background metric, localised between two Cauchy surfaces. In that situation both backgrounds must be globally hyperbolic and share the same Cauchy surface topology (the Cauchy surfaces must be related by orientation-preserving diffeomorphism). By contrast, the causal set framework would permit a perturbed causal set that modelled a change of topology relative to the unperturbed one, provided that suitable identifications can be made in the past and future infinity regions. 

When it is defined, the relative Cauchy evolution can be pulled back to the map on observables, as follows. Consider $F\in \Fcal^+_{Sol}(\Ccal)$, where $\Fcal^+_{Sol}(\Ccal)$ denotes the space of functionals on $\Ecal^+_{Sol}$. Define 
\be
\mathrm{rce}(F)(\phi)\doteq F(\mathtt{rce}(\phi))\,.
\ee
This map describes the change to the observable $F$ resulting from the perturbation to the underlying causal set.
 
\subsection{Peierls bracket}

\subsubsection{Tentative definition}

The next step is to define a Poisson structure on the space of observables on a fixed causal set. 
We do this using the method of Peierls \cite{Pei}.

Using the commutator function \eqref{df:E} (in analogy to \cite{Pei}), we define the following bracket on $\Ci(\Ecal(\Ccal),\CC)$
\begin{equation}\label{Pbracket}
\{F,G\}=\sum_{i=1}^{N} \sum_{j=1}^{N}\frac{\delta F}{\delta \phi_i} E^{ij} \frac{\delta G}{\delta \phi_j}\,,
\end{equation}
where we used the Euclidean inner product to raise one of the indices in $E_i^{\ j}$ (compare with Remark~\ref{rem:prod}).
To simplify the notation, we will drop the summation symbols, and use the condensed notation $\dfrac{\delta F}{\delta \phi_i}\equiv F_{,i}$, so the formula above becomes
\be
\{F,G\}=F_{,i} E^{ij}G_{,j} = (F^{(1)})^T E G^{(1)}\,,
\ee
using the index-free notation in the last expression.  
For linear observables this reduces to
\be
\{F_g,F_h\}=g^TE h\,,
\ee
where $F_g,F_h\in\Xcal(\Ccal)$.
\begin{prop}\label{prop:poisson:free}
	The bracket \eqref{Pbracket} is a Poisson bracket, in particular, it satisfies the Jacobi identity: for any $ F,G,H\in\Ci(\Ecal(\Ccal),\CC) $:
	\begin{equation}\label{jacobi}
	\{F,\{G,H\}\}+\{G,\{H,F\}\}+\{H,\{F,G\}\}=0.
	\end{equation}
\end{prop}
\begin{proof}
	The argument is standard but we give it for completeness and for comparison with a later result. 
	The antisymmetry is obvious from the definition of $E$. It remains to prove the Jacobi identity. 
Expanding \eqref{jacobi} we find:
	\begin{equation}\label{key1}
	F,_iE^{ij}(G,_kE^{kl}H,_l),_j+G,_iE^{ij}(H,_kE^{kl}F,_l),_j
		+H,_iE^{ij}(F,_kE^{kl}F,_l),_j =0
	\end{equation}
	of which the first term equals  
	\begin{equation}\label{key2}
	F,_iE^{ij}G,_{kj}E^{kl}H,_l+F,_iE^{if}G,_kE,_j^{kl}H,_l +F,_iE^{ij}G,_jE^{kl}H,_{lj}\,.
	\end{equation}
	Due to the antisymmetry of $ E $ and because $E$ is independent of the field, all the terms in \eqref{key1} cancel out, so the Jacobi identity follows.
\end{proof}
\subsubsection{Justification of the formula for the bracket}
We now discuss a sense in which \eqref{Pbracket} corresponds to a discrete version of the Peierls bracket \cite{Pei}, by showing that it represents the difference between suitably defined retarded and advanced responses of the field equation to linear perturbations, supposing that the unperturbed equation has a $k$-layer Cauchy problem. In the original work of Peierls \cite{Pei} the idea is to define a covariant bracket, using the Lagrangian formalism, by studying the response of a given observable (say $F$) to a perturbation of the action by another observable, $G$. One compares two situations: 
\begin{itemize}
\item $F$ is evaluated at the solution to the perturbed problem, which coincides with the free solution in the far past (retarded response) and 
\item $F$ is evaluated at the solution to the perturbed problem, which coincides with the free solution in the far future (advanced response)
\end{itemize}
The bracket of $F$ and $G$ is the linear term (in $G$) of the difference between the retarded and the advanced response of $F$ to $G$.

We start with analyzing the situation, where we add a source $\lambda g$ to the theory. 
Heuristically,  this means adding a linear functional $\lambda \Phi_g$ to the action, where $g$ is supported outside $C^-_k$ and both $\lambda$ and $g$ are real. We will implement this by a direct modification to the field equation. The idea of Peierls is to study the effect of having such a perturbation on the observables. Let $\phi$ be a solution to the non-perturbed field equation \eqref{eq:inhom} with Cauchy data $\phi^-$ at ${C_k^-}$ and let $\phi_{\lambda}$ be a solution to the perturbed equation
\begin{equation}\label{eq:Peierlslinpert}
P\phi_{\lambda}=K (f+{\lambda} g)+\phi^-,
\end{equation}
with the same Cauchy data. Now consider another linear observable $\Phi_h$ with $h$ real. The retarded response operator $D^+_{\Phi_g}$ is, in direct analogy to \cite{Pei}, a transformation of observables defined by
\begin{equation}
\left(D^+_{\Phi_g} \Phi_h\right)(\phi) =\lim\limits_{{\lambda}\rightarrow 0}\frac{1}{{\lambda}}(\Phi_h(\phi_{\lambda})-\Phi_h(\phi)) \,.
\end{equation}
In this case, Eq.~\eqref{eq:Peierlslinpert} gives
\begin{equation}
\left(D^+_{\Phi_g} \Phi_h\right)(\phi) =
\lim\limits_{{\lambda}\rightarrow 0}\frac{1}{{\lambda}}h^T( E^+(f+\lambda g)-E^+(f)) =h^TE^+ g\,,
\end{equation}
which is independent of the solution $\phi$, i.e., $D^+_{\Phi_g} \Phi_h$ is a constant functional.
Just as we defined the advanced Green function to be the transpose of the retarded version, we now 
define the advanced response by reversing the roles of the perturbation and the observable used to test the response
\begin{equation}
D^-_{\Phi_g} \Phi_h\doteq D^+_{\Phi_h} \Phi_g\,.
\end{equation}
With this definition, 
\begin{equation}
\left(D^-_{\Phi_g} \Phi_h\right)(\phi) = g^TE^+ h = h^T (E^+)^T g = h^T E^- g,
\end{equation}
so the Peierls bracket is
\begin{align}
\{\Phi_g,\Phi_h\}_{Pei}(\phi) &=D^+_{\Phi_g} \Phi_h-D^-_{\Phi_g} \Phi_h =h^T(E^{+}-E^{-}) g 
\notag \\
&= g^TEh\,,
\end{align}
in agreement with our definition \eqref{Pbracket}.

Turning to nonlinear perturbations and nonlinear observables, let $F\in\Ci(\Ecal(\Ccal),\CC)$ (some examples appear at the end of this subsection). Heuristically, perturbing the action by $\lambda F$ has the effect of perturbing the field equation by $\lambda \frac{\delta F}{\delta\phi_p}(\phi)$. Taking \eqref{eq:hom} as the starting point, this suggests that the interacting field equation should take the form:\footnote{Note that since $P\phi$ is heuristically \textit{minus} the variation of the free action, in order to implement the interaction, we have to \textit{subtract} the variation of $F$ on the left-hand side (or add it on the right-hand side). This differs from the usual convention used in \cite{DF02,Book,HR}, where $P=-\Box$, rather than $P=\Box$, as we assume in the present work.}
\be\label{eq:Interacting}
{P}\phi=\lambda K (\tfrac{\delta F}{\delta\phi}(\phi))+\phi^-\,,
\ee
where $F$ is supported away from the past infinity $C_k^-$. The matrix $K$ is there for consistency with linear perturbations (compare with \eqref{eq:Peierlslinpert}). This is an artifact of the way we choose to discretize the interaction term. As for the notion of support of $F$,  we adopt the following definition
\begin{equation}
\supp F \doteq \{p\in\Ccal| \exists \phi\in\Ecal(\Ccal),~\lambda\in\CC~\text{s.t.}~F(\phi+\lambda\delta_p)\neq F(\phi) \}\,.
\end{equation}
where $\delta_p(q)=\delta_{pq}$. This is a straightforward generalization of the notion used in continuum. An obvious consequence of this definition is that if $p\notin \supp F$, then $\frac{\delta F}{\delta\phi_p}(\phi)=0$. Conversely, we can express the support as
\be
\supp F=\bigcup_{\phi\in\Ecal(\Ccal)} \supp\left(\frac{\delta F}{\delta\phi}(\phi)\right)\,,
\ee
where $ \supp(\frac{\delta F}{\delta\phi}(\phi))$ consists of points $p$, for which $ \frac{\delta F}{\delta\phi_p}(\phi)$ is non-zero.

Let $F,G\in\Ci(\Ecal(\Ccal),\CC)$ be supported away from past infinity. Then the retarded response for the discretized wave equation \eqref{eq:inhom} is
\be
D^+_{F}(G)=G_{,i} (E^{+})^{ij}F_{,j}\,,
\ee
and defining $D^-_G(F)=D^+_F(G)$, as before, the bracket \eqref{Pbracket} agrees with the Peierls construction.

Note that in general $\frac{\delta F}{\delta\phi_p}(\phi)$ can be very non-local, i.e. it can depend on values of $\phi$ at points other than $p$. This motivates the following definition:
\begin{df}
	A functional $F\in\Ci(\Ecal(\Ccal),\CC)$ is called local if $\frac{\delta F}{\delta\phi_p}(\phi)$ is a function of $p$ and $\phi(p)$ (only).
\end{df}
Linear functionals considered in the previous section are obviously local. Other examples are local polynomials which are finite sums of terms with the form 
\begin{equation}\label{Local:pol}
F(\phi)=(\phi*\dots*\phi)^i g_i\,,
\end{equation}
using the Hadamard product~\eqref{eq:Hadamardproduct}. 
\subsection{Free Dynamics}
Let us define the Poisson algebra assigned to a causal set $\Ccal$ as
\be
\fP(\Ccal)\doteq (\Fcal,\{.,.\})\,,
\ee
which is the causal set counterpart of the \textit{off-shell} classical algebra in continuum pAQFT (equations of motion are not imposed). Typically, one would now quotient it by the ideal generated by the equations of motion, to obtain the on-shell algebra. The potential problem that arises in the causal set context is that the retarded Green function is the inverse to $P$ but its transpose is not (unless in very special cases). As a result, in general, $EP=(E^--E^+)P\neq 0$, which means that the Peierls bracket would not be well-defined on the quotient algebra.

Hence, instead of quotienting by the ideal generated by the equations of motion, we propose to quotient $\fP(\Ccal)$ by the ideal generated by functionals $F$ with the property
\be\label{quotinet}
E^{ij} F_{,j}\equiv 0\,,
\ee
denoting this quotient by $\widetilde{\fP}(\Ccal)$. Note that on $\widetilde{\fP}(\Ccal)$ the Poisson bracket $\{.,.\}$ is non-degenerate.

To see how the above quotient is related to implementing the dynamics, recall that in continuum we have the exact sequence \cite{Baer}:
\be
0\rightarrow\Dcal(M)\xrightarrow{P}\Dcal(M)\xrightarrow{E} \Ecal_{\rm sc}(M)\xrightarrow{P}  \Ecal_{\rm sc}(M)\,,
\ee
where $M$ is a globally hyperbolic spacetime, $P$ is a normally hyperbolic operator, $\Dcal(M)$ and $\Ecal_{\rm sc}(M)$  are space of functions with compact and spatially compact support, respectively. We also know that the space of linear observables is isomorphic to $\Dcal(M)$ by means of
\be
\Dcal(M)\ni g\mapsto \Phi_g;\ \Phi_g(\phi)=\int \phi(x)g(x) d\mu
\ee
Hence quotienting the algebra of regular functionals by the ideal generated by elements of the form $\Phi_{Pf}$, $f\in \Dcal(M)$ is the same as quotienting by the ideal generated by linear observables $\Phi_g$ with the property $ g\in \ker E$.

This result extends to more singular functionals by continuity. Clearly, our condition \eqref{quotinet} is the causal set analogue of quotienting by the kernel of $E$. This condition is also the natural generalization of the condition proposed in \cite{Sorkin} for linear observables. Sorkin argues that the kernel of $E$ in the causal set situation is typically very small. There are many very small eigenvalues of $E$, but only a few of them are exactly 0. This leads Sorkin to conclude that the equations of motion on a causal set can be implemented only in approximate sense \cite{Sorkin}. We hope  to address this point in future work.

\subsection{Interacting theory}
Next we want to introduce the interaction. We will use the framework proposed in \cite{DF02} and further developed in \cite{HR}.
\subsubsection{Interacting and linearized interacting equations of motion}
Let $V\in\Ci(\Ecal(\Ccal),\CC)$, where $|\Ccal|=N$  and let ${\lambda}$ be the coupling constant. We work perturbatively, so the space of observables is now extended to include formal power series in the coupling constant ${\lambda}$, i.e. it becomes $ \Fcal(\Ccal)[[{\lambda}]]$. 
The interacting field equations are given by \eqref{eq:Interacting}, which we can also write as
\be
{P}\phi-\lambda K (V^{(1)}(\phi))=\phi^-\,.
\ee
The interacting field equations linearized about $\phi\in\Ecal(\Ccal)$, are
\be\label{eq:inter:linear}
P\psi-{\lambda} KV^{(2)}(\phi) \psi=\psi^-\,.
\ee
where $V^{(2)}(\phi)$ is an $N\times N$ matrix with components
\be
(V^{(2)}(\phi))_{ij}=V_{,ij}(\phi) \,.
\ee

\subsubsection{Interacting Poisson bracket}
The prescription for the Poisson bracket  of the interacting theory with the interaction $\lambda V$ is given by
\be\label{Pbracket:int}
\left\{G,H\right\}_{{\lambda} V}\doteq G_{,i} E_{{\lambda} V}(\phi)^{ij} H_{,j}\,,
\ee
where $E_{{\lambda} V}(\phi)=(E^+_{{\lambda} V}(\phi))^T-E^+_{{\lambda} V}(\phi)$, and $E^{+}_{{\lambda} V}(\phi)$ is the retarded Green function for the interacting linearized field equations \eqref{eq:inter:linear}. Starting from the free Green function $E^{+}$, we construct the interacting one using the Neumann series:
\be
E^{+}_{{\lambda} V}=E^{+}+\sum_{n=1}^\infty{\lambda}^n E^{+} \left(V^{(2)} E^{+}\right)^n\,.
\ee
\begin{prop}\label{prop:poisson:int}
	The bracket \eqref{Pbracket:int} is a Poisson bracket, in particular, it satisfies the Jacobi identity: for any $ F,G,H\in\Ci(\Ecal(\Ccal),\CC) $, one has
	\begin{equation}\label{jacobi:int}
	\{F,\{G,H\}\}+\{G,\{H,F\}\}+\{H,\{F,G\}\}=0.
	\end{equation}
\end{prop}
\begin{proof}
The proof is analogous to that of Proposition~\ref{prop:poisson:free}. The only difference is in the proof of the Jacobi identity. Expanding the first term in \eqref{jacobi:int}, we obtain 
	\begin{equation}\label{eq:jacobi:terms}
	F_{,i}E_{\lambda V}^{il}G_{,jl}E_{\lambda V}^{jk}H_{,k}+
	F_{,i}E_{\lambda V}^{il}G,_j(E_{{\lambda} V})_{,l}^{jk}H_{,k} 
	+F_{,i}E_{\lambda V}^{il}G_{,j}E_{\lambda V}^{jk}H_{,kl}.
	\end{equation} 
	Due to the antisymmetry of $E_{\lambda V}$, only the central term of the expansion of each bracket remains as the others cancel across all three expanded brackets. The derivatives of the retarded Green function are:
		\begin{equation}\label{eq:Green:var}
		(E^{+}_{\lambda V})^{jk}_{,l}=\lambda(E^{+}_{{\lambda} V})^{jm}V_{,lmn}(E^{+}_{\lambda V})^{nk}\,,
		\end{equation}
		where $V$ is the interaction term. 	Inserting this into \eqref{eq:jacobi:terms} 
		and expanding each $E_{\lambda V}=(E^+_{\lambda V})^T-E^+_{\lambda V}$, we obtain twelve terms altogether. These cancel in pairs due to the antisymmetry of $E_{\lambda V}$ and the fact that $V,_{lmn}$ is symmetrical with respect to its indices (see the Appendix~B of~\cite{Jakobs} for more details of the proof).
\end{proof}
We can also introduce the retarded M{\o}ller map, which maps solutions to free discretised retarded field equations to solutions to interacting field equations: \eqref{eq:Interacting}
\be
\label{YF equation}
{\mathtt r}_{{\lambda} V}(\phi)=\phi+\la E^{+}V^{(1)}({\mathtt r}_{\la V}(\phi))\,,
\ee
Its inverse is given by
\be
\label{inverse Moller}
{\mathtt r}^{-1}_{\la V}(\phi) = \phi - \la E^{+} V^{(1)}(\phi) \,.
\ee
The retarded M{\o}ller map on configurations induces the corresponding map on observables, 
\be
(r_{\la V}F)(\phi)\doteq F\circ {\mathtt r}_{\la V}(\phi)\,,\\
\ee
where $F\in \Fcal(\Ccal)[[{\lambda}]]$.  
Analogously to the continuum case, the Peierls bracket \eqref{Pbracket:int} satisfies:
\be
\{F,G\}_{{\lambda} V}=r_{{\lambda} V}^{-1}\{r_{{\lambda} V}F,r_{{\lambda} V}G\}\,.
\ee
\section{Quantum Theory}
So far, we have constructed a Poisson algebra of observables for the free and interacting classical field theory on a causal set. We now  pass to the quantum theory using the method of deformation quantization, following ideas applied in the context of continuum field theory by \cite{DFloop,DF,BDF}. The main idea behind this approach is to phrase the problem of quantization as the mathematical problem of finding a non-commutative $\hbar$-dependent product $\star_\hbar$, such that, given a Poisson algebra (with a commutative product $\cdot$ and Poisson bracket $\{.,.\}$):
\begin{align*}
F\star_\hbar G&\xrightarrow{\hbar\rightarrow 0} F\cdot G\,,\\
\tfrac{1}{i\hbar}\left(F\star_\hbar G-G\star_\hbar F\right) & \xrightarrow{\hbar\rightarrow 0} \{F, G\}\,.
\end{align*}
In our case, $F$ and $G$ are smooth functions on the configuration space $\Ecal(\Ccal)$. Furthermore, the $\star_\hbar$ product is required to be of the form:
\be\label{eq:general:star}
F\star_\hbar G=\sum\limits_{n=0}^\infty\hbar^n B_n(F,G)\,,
\ee
where $B_n$ are differential operators on $\Ecal(\Ccal)$. For simplicity, we focus on finite causal sets, so that for a causal set of size $N$, $\Ecal(\Ccal)\cong \RR^N$ and the differential calculus on it is well understood. We can then express
\begin{equation}\label{eq:Bn}
B_n(F,G)=\sum_{i_1=1}^N\dots\sum_{j_n=1}^N  \frac{\delta^n F}{\delta \phi_{i_1}\dots \delta \phi_{i_n}}\  (B_n)^{i_1,\dots,i_n,j_1,\dots,j_n}\   \frac{\delta^n G}{\delta \phi_{j_1}\dots \delta \phi_{j_n}}
\end{equation}
The higher orders in $\hbar$ present in \eqref{eq:general:star} distinguish this approach from the Dirac quantization, so one avoids contradiction with the Groenewald-van Hove no-go theorem \cite{Groenewold,vanHove}.

Deformation quantization has the advantage that the construction of the algebra of observables can be done completely abstractly, without the need for existence of a distinguished state (e.g., a vacuum) and without invoking Fock space methods. At a later stage, one can then seek suitable states on the abstract algebra and use these to form Hilbert space representations by the GNS construction. 

The choice of states has both a mathematical and a physical aspect. There is a precise mathematical definition of a state, as a positive normalised linear functional on a $*$-algebra (and this can be extended to algebras of formal power series, as will be described below). However not all such linear functionals need qualify as physically relevant. For QFT in continuum curved spacetimes it is known that it is not possible to single out a unique distinguished state that is locally and covariantly determined by the geometry, assuming certain additional physically motivated assumptions -- see~\cite{spass1} for a formal proof and \cite{Fewster-artofstate:2018} for a review. Nonetheless, there are circumstances in which a distinguished global state (or distinguished `in' and `out' states) with good properties (specifically, the Hadamard condition) may be determined~\cite{DMP09,GerWro:2017,DerSiems:2019}. One proposal for a global geometrically determined state, arising from the causal set programme, is the Sorkin-Johnston (SJ) state \cite{Johnston,SJ12}. In the continuum this is known to have certain problems; in particular, it generally fails to be Hadamard~\cite{FV12,FV13}.
Nonetheless, SJ states retain interest as a specific construction of a state where particular examples are otherwise sparse; furthermore, there are softened versions~\cite{BrumF14,FL15,FrancisThesis} of the SJ construction in which the Hadamard property is restored at the price of losing uniqueness. 

As the definition of Hadamard states centres on the UV behaviour of their $n$-point functions, it may seem that these problems are vitiated in discrete spacetimes. This is not quite so, because ideally one would like to understand the class of states that can have Hadamard continuum limits; and if the causal set has infinitely many elements then there is still the possibility that different states could yield inequivalent representations of the CCR algebra. The situation is of course better still in the case of finite causal sets, our main focus, where the Stone-von Neumann theorem ensures that all sufficiently regular representations are unitarily equivalent up to multiplicity. In this situation any pure state will lead to the same Hilbert space representation. Therefore the SJ state is a valid starting point for a more refined discussion. As we will show, it is mathematically simple to describe, and closely related to our choice of the Euclidean inner product in Sec.~\ref{sec:kinematics}. 
Alternative inner products produce states that may be seen as precursors of the softened SJ states of~\cite{BrumF14,FL15,FrancisThesis}.  See Remark~\ref{rem:states} for some further comments.

A great advantage of the algebraic viewpoint is that it is much more straightforward
to introduce interactions than in constructions based on Fock space. Applying the ideas of \cite{BDF,Vienna,HR}, we will show how to pass to the interacting theory, using further deformation of the non-commutative product $\star_\hbar$.

\subsection{Construction of the Quantum Algebra}\label{sec:free}
\subsubsection{Exponential products}
Deformation quantization of the classical theory starts with the free theory, i.e. a linearized wave equation \eqref{eq:hom} and its retarded Green function $E^+$. From this we obtain $E$ and the Peierls bracket. Let us for the moment restrict ourselves to the subspace of $\Fcal(\Ccal)$ that consists of smooth functionals $F$ with the property that there exists $N\in\NN$ such that $F^{(n)}(\ph)=0$ for all $n>N$, $\ph\in\Ecal(\Ccal)$. We call such functionals \textit{polynomial} and denote the corresponding vector space by $\Fcal_{\pol}(\Ccal)$.
This space 
can be equipped with various types of noncommutative product, of which we will describe two. First, the \emph{Moyal--Weyl product} is
\be
\label{exponential form}
F\star G \doteq m\circ e^{\frac{1}{2}i\hbar D_E}(F\otimes G)  \,,
\ee
where $F,G\in\Fcal_{\pol}(\Ccal)$,
 $m$ is the multiplication on $\Fcal_\pol(\Ccal)$ induced by pointwise multiplication of functionals in $\Fcal_{\pol}(\Ccal)$ and for a given $N\times N$ matrix $M$,
\be
D_M\doteq M_{ij} \frac{\delta}{\delta\phi_i} \frac{\delta}{\delta\phi_j}\equiv
\left<M,\tfrac{\delta}{\delta \phi}\otimes \tfrac{\delta}{\delta \phi}\right>\,,
\ee
maps $\Fcal_{\pol}(\Ccal)^{\otimes 2}\to \Fcal_{\pol}(\Ccal)^{\otimes 2}$. With the appropriate choice of units, $\hbar$ is just a number and can be set equal to 1.
 We obtain a non-commutative algebra $\fA(\Ccal)\doteq (\Fcal_{\pol}(\Ccal),\star)$, which, as in the classical case, is the analogue of the continuum off-shell algebra.
Second, the \emph{Wick product} is defined by 
\be\label{Wick}
F\star_H G \doteq m\circ e^{\hbar D_W}(F\otimes G)  \,,
\ee
where 
\be\label{eq:W:decomp}
W=\frac{i}{2}E +H\,,
\ee
is a complex hermitian matrix that has the physical interpretation of the two-point function of a quasifree state on $\fA(\Ccal)$. Denote $\fA_H(\Ccal)\doteq (\Fcal_{\pol}(\Ccal),\star_H)$.

We require $W$ to have the following properties, which then constrain $H$:
\begin{enumerate}[label={\bf W\arabic*)}]
\item  $E=2\operatorname{Im}\,W$, i.e., $H=\operatorname{Re}\, H$ (recall that $E$ is real by definition).
\item $W$ is a positive definite matrix, meaning that $f^\dagger  W f\geq 0$, where $f^\dagger$ is the hermitian conjugate of $f\in\CC^N$.
\item $\ker W\subset \ker E$ (a proxy for $W$ solving the equations of motion)
\end{enumerate}
This is almost the same as in the continuum, but modifying the condition that $W$ solves the equations of motion. This is related to the general difficulty with going on-shell discussed before.
 \begin{rem}
	Note that both $\star$ and $\star_H$ are of the form \eqref{eq:general:star} with $(B_{n})^{i_1,\dots,i_n,j_1,\dots,j_n}=(i/2)^nE_{i_1j_1}\dots E_{i_nj_n}$ for the former and $(B_{n})^{i_1,\dots,i_n,j_1,\dots,j_n}=W_{i_1j_1}\dots W_{i_nj_n}$ for the latter.
\end{rem}
Physically, passing from $\star$ to $\star_H$ corresponds to normal-ordering with respect to the quasifree state determined by $W$. In deformation quantization, the products $\star$ and $\star_H$ are regarded as equivalent, because they are related by a \textit{gauge transformation} $\alpha_{H}: \fA(\Ccal) \rightarrow \fA_H(\Ccal)$, which is given by
\be
\alpha_{H} \doteq e^{\frac\hbar2\Dcal_H}\,,
\ee
where 
$$\Dcal_H(F) \doteq H^{ij} F_{,ij} \equiv \left<H,\tfrac{\delta^2 F}{\delta\phi^2}\right>\,,$$
More explicitly, we can write
\be
F\star_H G=\alpha_H(\alpha_H^{-1} F\star \alpha_H^{-1} G )\,,
\ee
and we identify $\alpha_H^{-1}(F)\equiv \no{F}\nolimits_H$ as the normal (Wick) ordering operation. Applying $\alpha_{H}^{-1}$ to a local functional $F\in \Fcal_{\loc}(\Mcal)$ in continuum is analogous to normal ordering using the point-splitting prescription.
\begin{exa}[Minkowski spacetime]
Consider the example of continuum free scalar field theory on Minkowski spacetime $\M$. Let $\Phi_f$, $\Phi_g$ be smeared fields, as defined by \eqref{eq:smeared:cont}. For the Klein-Gordon operator $P=\Box+m^2$, there exist the unique retarded and advanced Green functions  $E^{\pm}$ and their difference is the commutator function $E=E^--E^+$. With the star product taken to be $\star_H$, for any choice of the symmetric part $H$, we have 
\[
	[\Phi_f,\Phi_{f'}]_{\star_H}=\Phi_{f}\star_H \Phi_{f'}-\Phi_{f'}\star_H \Phi_{f}=i\hbar\left<E,f\otimes f'\right>\,.
	\]
Formally, we can consider $\Phi_x\doteq \Phi(\delta_x)$, where $\delta_x$ is the Dirac delta supported at some $x\in \M$, and we find:
\[
[\Phi_x,\Phi_y]_{\star_H}=i\hbar E(x,y)\,.
\]
Now fix an inertial coordinate system and consider the $t=0$ Cauchy surface. Let ${\bf x}$ and ${\bf y}$ denote spacelike vectors on this surface. From properties of $E$ follows that:
\begin{align*}
[\Phi_{(0,{\mathbf{x}})},\Phi_{(0,{\mathbf{y}})}]_\star&=\Delta(0,{\mathbf{x}}; 0,{\mathbf{y}})=0\\
[\Phi_{(0,{\mathbf{x}})},\dot{\Phi}_{(0,{\mathbf{y}})}]_\star&=\partial_{y^0}\Delta(0,{\mathbf{x}}; 0,{\mathbf{y}})=i\hbar\delta(\mathbf{x}-\mathbf{y})\,,
\end{align*} 
where dot denotes the time derivative. These are the standard equal-time commutation relations, so we have recovered the usual formulas from the deformation quantization approach.
\end{exa}
In order to find a specific choice of $W$, we will follow the ideas of \cite{Johnston,SJ12} and take $W$ as the Sorkin-Johnston (SJ) two-point function. We recall, that according to \cite{Sorkin}, $W$ is the unique $N\times N$ matrix satisfying the following properties:
\begin{enumerate}[label={\bf SJ\arabic*)}]
	\item $W-\overline{W}=i E$, where bar denotes the complex conjugation,
	\item $W\geq 0$,
	\item $\overline{W}W=0$\,.
\end{enumerate}
It was shown in \cite{Sorkin} that the unique $W$ satisfying the axioms above is given by
\be
W=\frac{1}{2}(iE +\sqrt{-E^2})\,,
\ee
where the square root is the unique positive semi-definite square root of the positive semi-definite matrix $(iE)^2=(iE)(iE)^\dagger$.
\begin{rem}\label{rem:states}
	Note that we have made implicit use of the Euclidean inner product
	in order to identify the 2-point function with a matrix. As already indicated in Remark~\ref{rem:prod}, this choice is to some extent arbitrary and could be changed. Since {\bf SJ3)} crucially depends on this choice, a different auxiliary inner product would produce a different 2-point function and hence \textit{a different state}. The significance of this fact becomes acute in the continuum. Consider the real scalar field on $(M,g)$ with $M=(-\tau,\tau)\times \Sigma$ being an ultrastatic slab spacetime. Choosing the inner product on $\Ci_c(M,\RR)$ to be the one induced by the volume form $d\mu_g$ implies that the unique $W$ solving {\bf SJ1)-3)} is the 2-point function of the SJ state which is known not to be Hadamard in general~\cite{FV12,FV13}. However, replacing $d\mu_g$ by $\frac{1}{\rho}d\mu_g$ produces $W_\rho$, which is a 2-point function of a \textit{Hadamard state}, if $\rho$ is an appropriately chosen smooth function on the interval $(-\tau,\tau)$, tending to zero at both endpoints (see \cite{FrancisThesis} for details).
	\end{rem} 
\subsubsection{Formal power series}\label{sec:formal power series}
Going beyond the polynomial observables requires some caution, since the power series defining the star product might not converge. One possibility is to consider analytic functions (e.g. exponentials, as discussed in the next section) or to extend the framework to allow also formal power series. The latter is actually necessary if we want to introduce interactions (see Section~\ref{sec:inter}).

Let $\Fcal(\Ccal)[[\hbar]]$ denote the vector space consisting of formal power series in the formal parameter $\hbar$, with coefficients in $\Fcal(\Ccal)$. Formulas \eqref{exponential form} and \eqref{Wick} can be easily adapted to this setting, but now we interpret $\hbar$ as a formal parameter rather than a number. The resulting algebras will be denoted by $\fA^{\hbar}\doteq (\Fcal(\Ccal)[[\hbar]],\star)$ and $\fA_H^{\hbar}\doteq (\Fcal(\Ccal)[[\hbar]],\star_H)$, respectively.

\subsection{Weyl Algebra}
The algebra of observables $\fA(\Ccal)$  introduced in the previous section can be equipped with a topology that makes it into a topological unital $*$-algebra. Such algebras are typically represented in Hilbert spaces by unbounded operators. If we want to work with bounded operators instead, a suitable candidate for the algebra of observables is provided by the Weyl algebra, defined by exponentiating linear functionals. In this section we treat $\hbar$ as a number, rather then a formal parameter.

Recall that elements of  $\RR^N$ are identified with  linear real observables on a causal set $\Ccal$ of size $N$ by means of \eqref{causet obs}. The Poisson bracket $\{.,.\}$ on the space of observables is given by \eqref{Pbracket}, so for $g,h\in \RR^N$ we have
\be
\{\Phi_g,\Phi_h\} =\left<g,E h\right>\doteq \sigma(g,h)\,.
\ee

\begin{df}
	Each linear real observable $\Phi_g$, $g\in\RR^N$, defines a 
	Weyl functional $\Wcal(g)\in\Ci(\Ecal(\Ccal),\CC)$ by $\Wcal(g)=e^{i \Phi_g}$.
\end{df}

\begin{prop}
	The Weyl functionals satisfy the Weyl commutation relations 
	\be
	\Wcal(g)\star \Wcal(\tilde{g}) = e^{-\frac{i\hbar}{2}\sigma(g,\tilde{g})}\Wcal(g+\tilde{g})\,,
	\ee
	with respect to the star product, and $\Wcal(g)^*=\Wcal(-g)$. 
	\begin{widetext}
	\begin{proof}
		This is a simple computation. The functional derivative of the operators in the direction of an arbitrary field in the configuration space $ h\in\Ecal(C) $ is:
		\begin{align}
		\left\langle(\Wcal(g))^{(1)}(\phi),h\right\rangle&=\left.\dfrac{d}{d{\lambda}}\exp\left(i\sum_{j=1}^{N}g_i(\phi_j+{\lambda} h_j)\right)\right|_{{\lambda}=0}=\left(i\sum_{j=1}^N g_jh_j\right)\Wcal(g)(\phi)
		\end{align}
		hence:
		\begin{equation}\label{nthweyl}
		\left\langle(\Wcal(g))^{(n)}(\phi),h^{\otimes n}\right\rangle=\left(i\sum_{j=1}^N g_jh_j\right)^n\Wcal(g)(\phi).
		\end{equation}
		Therefore, the following formula is obtained from the star product:
		\begin{align}
		\Wcal(g)\star \Wcal(\tilde{g})&=\sum_{n=0}^{\infty}\frac{\hbar^n}{n!}\left\langle(\Wcal(g))^{(n)},\left(\frac{iE}{2}\right)^{\otimes 	n}(\Wcal(\tilde{g}))^{(n)}\right\rangle \nonumber \\
		&=\sum_{n=0}^{\infty}\left(\frac{i\hbar}{2}\right)^n\frac{(-1)^n}{n!}\left(\sum_{i,j=1}^Ng_iE^{ij}\tilde{g}_j\right)^n\Wcal(g+\tilde{g}) \nonumber \\
		&=\exp\left(-\frac{i\hbar}{2}\sum_{i,j=1}^Ng_iE^{ij}\tilde{g}_j\right)\Wcal(g+\tilde{g}) \nonumber \\
		&=e^{-\frac{i\hbar}{2}\{F_g,F_{\tilde{g}}\}}\Wcal(g+\tilde{g})= e^{-\frac{i\hbar}{2}\sigma(g,\tilde{g})}\Wcal(g+\tilde{g})\,,
		\end{align}
		as required.
	\end{proof}
\end{widetext}
\end{prop}
We may now introduce the Weyl $C^*$-algebra for the free scalar field on a causal set $\Ccal$. For details see for example \cite[section 8.3.5]{DerGer13} or \cite[section 14.2]{Moretti}. Consider the non-separable Hilbert space $\Hcal=L^2(\RR^N,d\mu)$ of square integrable functions with the counting measure $\mu$ on $\RR^N$. This can also be identified with $l^2(\RR^N)$, the space of square-summable sequences indexed over $\RR^N$, because any element of $\Hcal$ may be written 
\be
\sum_{g\in\RR^N} c_g e_g\,,\quad\textrm{with}\quad \sum_{g\in\RR^N} |c_g|^2<\infty\,,
\ee
where $\{e_g\}_{g\in \RR^N}$ is the orthonormal basis for $\Hcal$ given by $(e_g)_h=\delta_{gh}$.  
The representation of Weyl generators is given by
\be
(\pi(\Wcal(h))a)_g \doteq e^{-\frac{i\hbar}{2} \sigma(g,h)}a_{g+h}\,. 
\ee
for any $a\in l^2(\RR^N)$. One may easily check that the Weyl relations are fulfilled, by explicit computation. Using this representation, we define a $C^*$-norm $\|.\|$ on the generators (by taking the corresponding operator norm as operators on $\Hcal$). We are now ready to define the  Weyl $C^*$-algebra.
\begin{df}
	The Weyl $C^*$-algebra is generated by the operators $ \{\Wcal(g)\}_{g\in\RR^N} $ and completed with respect to the $C^*$-norm $\|.\|$ 
\end{df}

\subsection{States and the GNS representation}
Within the framework of algebraic quantum theory, a physical system is described by the algebra of observables associated with it. The abstract algebra may be linked to the standard formulation of quantum theory by means of a Hilbert space representation. If we start with a $C^*$-algebra, we can represent it by bounded operators. For a general topological unital $*$-algebra, we have to work with unbounded operators as well.

Choosing a Hilbert space representation is equivalent to choosing an algebraic state, by virtue of the Gelfand-Naimark-Segal (GNS) construction.

\begin{df}
	Let $ \mathfrak{A} $ be a  topological unital $*$-algebra, then an algebraic state is a linear functional $ \omega:\mathfrak{A}\to\CC $ such that:
	\begin{equation}\label{key}
	\omega(a^*a)\geq0 \quad \forall ~ a\in\fA, \qquad \omega(\1)=1.
	\end{equation}
\end{df}

\begin{thm}[GNS]
	Let $\omega$ be a state on a unital $*$-algebra $\fA$. Then there exists a representation $\pi$ of the algebra by linear operators on a dense subspace $\Kcal$ of some Hilbert space $\HH$ and a unit vector $\Omega \in \Kcal$, such that
	\be
	\omega(A) = (\Omega, \pi(A)\Omega)\, ,
	\ee
	and $\Kcal = \{\pi(A)\Omega, A \in \fA\}$.
\end{thm}

For the details of the proof, which is constructive and builds the Hilbert space from the algebra, see for example \cite[Section 2.1.3]{Book} or~\cite[\S 2.3]{FewRej2019}. Let us now record some general facts about states in the pAQFT framework.

Firstly, we establish that evaluation at the zero vector is a state on $\fA_{H}(\Ccal)$. The corresponding result in the continuum case is known, but to the best of our knowledge there is no complete proof written down anywhere in the literature, so for completeness we provide it here as well.
\begin{df}
	Let $\Mcal=(M,g)$ be a globally hyperbolic manifold, $\Ecal\doteq \Ci(M,\RR)$, $\Dcal^{\CC}_n\doteq \Ci(M^n,\CC)$. Define the space of regular polynomials $\Fcal_{\rm pol}(\Mcal)$ as the space of functionals on $\Ecal$ of the form:
	\be\label{eq:f:reg}
	F(\ph)=F_0+\sum_{k=1}^{N} \left<\ph^{\otimes k},f_k\right>\,,
	\ee
	where $F_0\in \CC$, $\ph\in\Ecal$, $f_k\in\Dcal^{\CC}_k$ and $\left<.,.\right>$ denotes the usual pairing induced by integrating with the invariant volume form $d\mu_g$ over the whole $M^k$.
\end{df}

\begin{prop}\label{prop:ev:zero:state}
	Let $\fA_H(\Mcal)=(\Fcal_{\rm pol}(\Mcal),\star_H)$ for some choice of a Hadamard function (by that we mean a bi-distribution satisfying the continuum version of {\bf W1)-W3)}, see \cite{Book} for the precise definition) $W=\frac{i}{2}E+H$. Set $\hbar=1$. The functional given by evaluation at zero
	\be
	\omega_{0}(F)\doteq F(0)\,,\qquad F\in\fA_H(\Mcal)\,,
	\ee
	is a quasi-free Hadamard state on $\fA_H(\Mcal)$ with 2-point function $W$.
\end{prop}
\begin{proof}
Take $F$ as in \eqref{eq:f:reg} and write (it is useful to keep $\hbar$ explicit at this stage)
\begin{align}
\omega_{0}(F^*\star_H F)&= \sum_{n=0}^{\infty}\frac{\hbar^n}{n!} \left<(F^{(n)}(0))^*,W^{\otimes n} F^{(n)}(0)\right> \notag\\
&= |F_0|^2+ \sum_{k=1}^{N}\hbar^k k! \left<\overline{f}_k,w_k f_k\right>\,,
\end{align}
where $w_k$ is a distribution in $\Dcal'_n$ defined by the following distributional kernel:
\begin{equation}
 w_k(x_1,\dots,x_k,y_1,\dots,y_k)\doteq  ((\Phi_{x_1},\dots \Phi_{x_k})^*\star_H (\Phi_{y_1},\dots \Phi_{y_k}))(0)\,,
\end{equation}
where $\Phi_{x}$ is the evaluation functional at $x\in M$, i.e. for $\ph\in\Ecal$: $\Phi_{x}(\ph)\doteq \ph(x)$.

Hence for the positivity of $\omega_0$ it is sufficient to show that all $w_k$, $k\in \NN$ are positive type, i.e.,  $w(\bar{F},F)\ge 0$. We proceed by induction. First, note that for $k=1$ we have $w_1=W$, which is by assumption positive  type. We need to prove the induction step, i.e. assuming $w_{n-1}$ is positive definite, we want to show that $w_n$ is positive type.

Our proof  is similar to the one used in \cite{DerGer13} for states on the Weyl algebra, but is more general. First we recall the Schur product theorem about positive semi-definite  matrices:
if $A\geq 0$ and $B\geq 0$ ($A$, $B$ are positive semi-definite), then their Hadamard product $A*B$ is also positive semi-definite.

\begin{widetext}
Let $\underline{x}\equiv (x_1,\dots,x_n)$, $\underline{y}\equiv (y_1,\dots,y_n)$. We express $w_n$ as
\be\label{eq:fact}
w_n(\vx,\vy)=\sum_{i,j=1}^n (\Phi_{x_i}^*\star_H \Phi_{y_j})(0) ((\Phi_{x_1}\cdot\ldots\cdot\widehat{\Phi_{x_i}}\cdot\ldots\cdot\Phi_{x_n})^*\star_H (\Phi_{y_1}\cdot\ldots\cdot\widehat{\Phi_{y_j}}\cdot\ldots\cdot\Phi_{y_n}))(0)\,,
\ee
where $\widehat{\ }$ indicates that the given symbol is omitted. Let $f=f_1\cdot\ldots\cdot f_n\in\Dcal_n$, where $f_i\in\Dcal$, $i=1,\dots,n$. 
Using \eqref{eq:fact} we obtain
\be
\left<\bar{f},w_n f\right>=\sum_{i,j=1}^{n} a_{ij} b_{ij}\,,
\ee
where 
\begin{align}
a_{ij}&\equiv\left<\bar{f}_i,W f_j\right>\,,\nonumber\\ b_{ij}&\equiv \left<(f_1 \cdot\ldots\cdot \widehat{f_i} \cdot\ldots\cdot f_n)^*, w_{n-1} (f_1 \cdot\ldots\cdot \widehat{f_j} \cdot\ldots\cdot f_n)\right>
\end{align}
\end{widetext}
Define $n\times n$ matrices $A\equiv [a_{ij}]$ and $B\equiv [b_{ij}]$. These are both positive semi-definite. To see this, we consider $\underline{\lambda}\in\CC^n$ and define
\be
\tilde{f}_{\underline{\lambda}}\doteq \sum_{i=1}^{n}\lambda_i f_i\,,\quad \hat{f}_{\underline{\lambda}}\doteq \sum_{i=1}^{n} \lambda_i f_1 \cdot\ldots\cdot \widehat{f_i} \cdot\ldots\cdot f_n\,.
\ee
It follows that
\be
\underline{\lambda}^\dagger A \underline{\lambda} = \left<\overline{\tilde{f}_{\underline{\lambda}}},W \tilde{f}_{\underline{\lambda}}\right>\geq 0\,,
\ee
since $W$ is positive semi-definite, and
\be
\underline{\lambda}^\dagger B \underline{\lambda} = \left<\overline{\hat{f}_{\underline{\lambda}}},w_{n-1} \hat{f}_{\underline{\lambda}}\right>\geq 0\,,
\ee
using the assumption in the induction step. By Schur product theorem the Hadamard product $A*B$ is also positive semi-definite and we note that 
\be
\left<\bar{f},w_n f\right>=\lambda_1^\dagger( A*B) \lambda_1\,,
\ee
where $\lambda_1=(1,\dots,1)$, so we conclude that
\be
\left<\bar{f},w_n f\right>\geq 0\,,
\ee
which proves the induction step.

It remains to show that $\omega_{0}$ is a quasifree state. This, however, is straightforward, since for $f_1,\dots,f_{2k}\in\Dcal$ the correlation function is given by 
\be
\omega_0(\Phi_{f_1}\star_H\cdots\star_H\Phi_{f_{2k}})=\sum_{G\in\mathcal{G}_{2k}} \prod_{e\in G}
W(f_{s(e)},f_{t(e)}),
\ee
where $f_i\in\Dcal$, $i=1,\dots,2k$, $\mathcal{G}_{2k}$ is the set of directed graphs with vertices labelled $1,\ldots,2n$, such that each vertex is met by exactly one edge and the source and target of each edge obey $s(e)<t(e)$. 

Correlation functions of an odd number of fields vanish, since all uncontracted factors of $\ph$ give zero after the evaluation.
\end{proof}

\begin{rem}
	We could replace $\Fcal_{\rm pol}(\Mcal)$ with a larger space of functionals, e.g. the \textit{microcausal functionals}~\cite{BF0}. The proof is then exactly the same, but the test function spaces $\Dcal_n$ are replaced by appropriate spaces of distributions satisfying given wavefront set conditions.
\end{rem}
Let us now state the discrete version of the Proposition~\ref{prop:ev:zero:state}. Again, we set $\hbar=1$.
\begin{prop}
	The functional given by evaluation at zero
	\be
	\omega_{0}(F)\doteq F(0)\,,\qquad F\in\fA_H(\Ccal)\,,
	\ee
	is a quasi-free state on $\fA_H(\Ccal)$, with  2-point function $W=\frac{i}{2}E +H$.
\end{prop}
\begin{cor}
	The functional given by
	\be
	\omega_{H,0}(F)\doteq  \al_{H}(F)(0)
	\ee
	is a state on $\fA(\Ccal)$ and if we take $W=\frac{i}{2}E +H$ to be that of the SJ state, then 
	\be
	\omega_{H,0}=\omega_{\rm SJ}\,.
	\ee
\end{cor}
\begin{proof}
	The 2-point function of $\omega_{H,0}$ is given by
	\be
	\omega_{H,0}(\Phi_f \star \Phi_g) =(\Phi_f\star_H \Phi_g)(0)= W=(\tfrac{i}{2}E+H)\,.
	\ee
\end{proof}
In particular, for Weyl generators, we obtain
\be
\omega_{H,0}(\Wcal(g))=e^{-\frac{\hbar}{2}\langle g,Hg\rangle}\,,
\ee
so $H$ is the \textit{covariance} of the state $\omega_{H,0}$.

As stated at the beginning of this section, passing
between $\fA$ and $\fA_{H}$ can be understood as normal ordering.  Hence, on one hand we can work with normal-ordered quantities $\alpha_{H}^{-1}(F)\equiv \no{F}\nolimits_{H}$, $\no{G}\nolimits_{H}$ within the algebra $\fA$ or with original functionals $F,G$ within the algebra $\fA_H$. Correlation functions are then computed using the rule:
\be
\omega_{H,0}(\no{F}\nolimits_H\star \no{G}\nolimits_H) =\omega_0(F\star_H G)=(F\star_H G)(0) \,.
\ee

Let us now discuss the generalization to the situation, where $\Fcal_{\pol}(\Ccal)$ is replaced with the space $\Fcal(\Ccal)[[\hbar]]$ of formal power series. 

We need the notion of states on the formal power series algebra $\fA^\hbar=(\Fcal[[\hbar]],\star)$. Condition \eqref{key} has to be understood in the sense of the formal power series. For $A=\sum_{n=0}^{\infty} \hbar^n A_n$ and $\omega=\sum_{n=0}^{\infty} \hbar^n\omega_n$, 
the normalization condition is that $\omega(\1)=\omega_0(\1)=1$.
We have
\begin{equation}
\omega(A^*A)=\omega_0(A_0^*A_0)+\hbar(\omega_0(A_1^*A_0+A_0^*A_1)+\omega_1(A_0^*A_0))+\dots\,,
\end{equation}
and, by definition, positivity for a formal power series means that the lowest non-vanishing term has to be positive (see \cite{BW98}), so if $\omega_0(A_0^*A_0)\neq 0$, then we require $\omega_0(A_0^*A_0)\geq 0$, i.e. $\omega_{0}$ is a state in the usual sense. Alternatively, one could use the stronger notion of positivity, proposed in \cite{DFqed}, where one says that a formal power series $b=\sum_{n=0}^{\infty} \hbar b_n$ is non-negative, if there exists a series $c=\sum_{n=0}^{\infty} \hbar c_n$, such that $b=c^*c$. This does not make any difference for what follows.

\subsection{Interacting theory}\label{sec:inter}
\subsubsection{Motivating the approach}\label{sec:motivation}
We finish this section with the construction of the interacting theory for a given interaction $V\in \Fcal(\Ccal)$. We use the framework of perturbative AQFT \cite{BDF,Vienna,Book}, where the interacting fields are constructed with the use of quantum M{\o}ller operators. The motivation comes from the interaction picture in quantum mechanics. Consider the continuum theory of the scalar field on Minkowski spacetime with the free Hamiltonain $H_0$ and let the interaction Hamiltonian begiven by $-\int_\Sigma  \normOrd{\Lcal_I(0,\mathbf x)}d\sigma$, where the integration goes over some fixed Cauchy surface and $\normOrd{\Lcal_I}$ is the normal-ordered Lagrangian density.

Heuristically, we would like to be able to use Dyson's formula for interacting fields, which reads:
\begin{equation}
\ph_I(x)=U(x^0,s)^{-1}\ph(x)U(x^0,s) =U(t,s)^{-1}U(t,x^0)\ph(x)U(x^0,s)\,,
\end{equation}
where $\ph(x)$ is the free field, $\ph_I(x)$ is the interacting field and the interacting time evolution operator is given by: 
\begin{widetext}
\begin{equation}
U_I(t,s)=1+\sum_{n=1}^\infty \frac{i^n\la^n}{n!} 
\int_ {([s,t]\times\mathbb R^3)^n}\!\!\!\! T(\normOrd{\Lcal_I(x_1)}\dots \normOrd{\Lcal_I(x_n)})  d^{4}x_1\dots d^4x_n\,,
\end{equation}
\end{widetext}
where $\la$ is the coupling constant,  $T$ denotes time-ordering and $\normOrd{\Lcal_I}$ is an operator-valued function given by
\[\Lcal_I(x)=e^{iH_0x^0}\normOrd{\Lcal_I(0,\mathbf x)}e^{-iH_0x^0}\,,
\]
Unfortunately, there are many problems with the above idea, already in Minkowski spacetime:
\begin{itemize}
	\item Typical Lagrangian densities, e.g. $\normOrd{\Lcal_I(x)}=\normOrd{\ph(x)^4}$ cannot be restricted to a Cauchy surface as operator-valued distributions. This is the source of the UV problem. On causal sets, the major problem with such quantities is the lack of a good analog of a Cauchy surface.
	\item There is a problem with taking the adiabatic limit, as the integral of the Lagrangian density over $\mathbf x$ does not exist if $\Sigma$ is non-compact. 
	\item The overall sum might not converge.
\end{itemize}
In pAQFT, the first two problems are addressed as follows: quantities like $\Lcal_I(0,{\bf x})$ are replaced by smeared ones, of the form: $V\equiv\int_{\M} f(x)\Lcal_I(x) d^4x$, where $f\in\Dcal(\M)$ plays the role of the adiabatic cutoff. Note that the expression we use is now fully covariant, so there is no need for singling out a Cauchy surface. The normal ordering is achieved by fixing a Hadamard function $W=\frac{i}{2}E +H$ and using the corresponding $\star_H$ product in the free theory. Finally, the time-ordered products corresponding to the above choice of a Hadamard function have to be constructed. In continuum, this is achieved through the Epstein-Glaser renormalization \cite{EG}, but on causal sets we are able to use a more direct method, as shown below. The problem with overall convergence cannot be addressed with our methods, so we will work with formal power series in $\hbar$. As an intermediate step, we will also use formal power series in the coupling constant $\lambda$. 

In pAQFT time-ordered products are not just formal expressions, but they stem from a binary product $\T$, so that:
\begin{equation}
T\left(\int_{\M} f_1(x)\Lcal_I(x) d^4x\dots \int_{\M} f_n(x)\Lcal_I(x) d^4x\right)
:= \int_{\M} f_1(x)\Lcal_I(x) d^4x\T\dots\T \int_{\M} f_n(x)\Lcal_I(x) d^4x\,.
\end{equation}
Constructing $\T$ will be one of the main tasks in the following section.
\subsubsection{$S$-matrix and interacting fields}

We start with the algebra $\fA(\Ccal)$, constructed in section \ref{sec:formal power series}, but we introduce a new formal parameter ${\lambda}$, which plays the role of the coupling constant. In this section $\fA^{\hbar,\lambda}(\Ccal)\equiv(\Fcal(\Ccal)[[\hbar,{\lambda}]],\star)$.  We fix a Hadamard function $W=\frac{i}{2}E +H$ and denote
 $\fA^{\hbar,\lambda}_H(\Ccal)\equiv(\Fcal(\Ccal)[[\hbar,{\lambda}]],\star_H)$. On this algebra there is a distinguished state given by evaluation at 0, i.e.
\be
\omega_0(F)\doteq F(0)\,.
\ee
 Different choices of $W$ lead to different, but \textit{isomorphic} algebras (all of them isomorphic to $\fA^{\hbar,\lambda}(\Ccal)$), each one with its own distinguished state, given by evaluation at 0. In each case, the 2-point function of this state is by definition $W$. 
 
 Since $\fA^{\hbar,\lambda}(\Ccal)$ and  $\fA^{\hbar,\lambda}_H(\Ccal)$ are related through the isomorphism $\alpha_{H}$, we obtain a family of states on   $\fA^{\hbar,\lambda}(\Ccal)$  labeled by $H$, i.e.:
\[
\omega_{H,0}(F)=\alpha_{H}(F)(0)\,.
\]
The 2-point functions of these states are given by:
\[
\alpha_{H}(\Phi_f\star\Phi_g)(0)=(\Phi_f\star_H\Phi_g)(0)=f^TWg\,.
\] 
If $W$ satisfies {\bf SJ1)-3)},  then $\omega_{H,0}$ is the SJ state on $\fA^{\hbar,\lambda}(\Ccal)$.

Now we are ready to introduce time-ordered products. To this end, we will need the Feynman propagator. In our framework, there is a ``Feynman-like'' propagator for every choice of $W$, so it is a state-dependent notion. We define it as
\be
\Delta^{\rm F}=\frac{i}{2}(E^++E^-)+H\,,
\ee
where $H$ is the symmetric part of the 2-point function $W$. In what follows, we will keep $H$ fixed and refer to $\Delta^{\rm F}$ simply as \textit{the Feynman propagator}.

Given $\Delta^{\rm F}$, we define the time-ordered product $\T$ by 
\begin{equation}\label{eq:timeord}
F\T G \doteq m\circ e^{\frac{\hbar}{2} D_{\Delta^{\rm F}}}(F\otimes G)=\sum_{n=0}^{\infty} \frac{\hbar^n}{n!} F_{,i_1\dots i_n} (\Delta^{\rm F})^{i_1j_1}\dots (\Delta^{\rm F})^{i_nj_n} G_{,j_1\dots j_n}  \,.
\end{equation}
We can also write it as
\[
F\T G=\TT (\TT^{-1} F\cdot  \TT^{-1} G)\,,
\]
where the product on the right-hand side is the usual point-wise product of functionals in $\Fcal(\Ccal)[[\hbar,\lambda]]$ and the operator $\TT$ is given by
\[
 \TT\doteq e^{\frac{\hbar}{2}\Dcal_{\Delta^{\rm F}}}\,.
\]
\begin{rem}
	On Minkowski spacetime,	the operator $\TT$ formally corresponds to a path integral involving a `Gaussian measure' with covariance $i\hbar\Delta^\mathrm{F}$, i.e.
	\be\label{path integral}
	\TT F(\ph)\stackrel{\mathrm{formal}}{=} \int F(\ph-\phi)  d\mu_{i\hbar\Delta^{\mathrm{F}}}(\phi)\ . 
	\ee
	Therefore, one can think of the pAQFT formalism as a way to make path integrals and other formulas used in perturbative QFT rigorous. We hope that this statement can be made more precise in the context of causal sets, where the path integral has better chances of being well-defined.
\end{rem}

Using the time-ordered product, we introduce \textit{the formal S-matrix} for the interaction $V$ and coupling constant ${\lambda}$. It is given by
\be
\Scal({\lambda} V)\doteq e_{\TT}^{\frac{i}{\hbar} {\lambda} V}=\sum_{n=0}^{\infty}  \frac{{\lambda}^ni^n}{\hbar^n n!} \underbrace{V\T\dots \T V}_n\,.
\ee
Next, we define the interacting fields.
For a classical observable $F\in \Fcal(\Ccal)$, the corresponding quantum interacting field is given by $R_{{\lambda} V}(F)$, where $R_{{\lambda} V}$ is the retarded quantum M{\o}ller operator defined by 
\begin{align}
R_{{\lambda} V}(F) &\doteq -i\hbar \frac{d}{d\mu} \Scal({\lambda} V)^{-1}\star_H \Scal({\lambda} V+\mu F)\big|_{\mu=0} \notag\\
&= \Big( e_{\TT}^{\frac{i}{\hbar} {\lambda} V}\Big)^{-1}\star_H \Big( e_{\TT}^{\frac{i}{\hbar} {\lambda} V}\T F\Big)\,.
\end{align}
Note the analogy of this formula to the Dyson series mentioned in section~\ref{sec:motivation} as the motivation for the pAQFT approach.

We can also use the M{\o}ller operator to deform the free star product and obtain the interacting one, using the formula:
\be\label{starint}
F  \starHint G\doteq R^{-1}_{{\lambda} V}(R_{{\lambda} V}(F)\star_H R_{{\lambda} V}(G))\,.
\ee
This way we obtain the interacting algebra $\fA_H^{\rm int}(\Ccal)\doteq (\Fcal(\Ccal)[[\hbar,{\lambda}]],\starHint)$. Given a state $\omega$ on the free algebra, we can construct the state $\omega_{\rm int}$ on  $\fA_H^{\rm int}(\Ccal)$ using the pullback:
\be
\omega_{\rm int}(F)\doteq \omega\circ R_{{\lambda} V}(F)\,,
\ee
where $F\in \Fcal(\Ccal)[[\hbar,{\lambda}]]$. The natural choice of a state in this context is $\omega_0$ for the free theory. Next, we want to choose observables. The first natural candidate is the smeared interacting field itself, i.e. 
\[
\Phi^{\rm int}_f\doteq R_{{\lambda} V}(\Phi_{f})
\]
The $n$-point correlation function of smeared interacting fields is given by:
\begin{align}
\omega_n^{\rm int}(g_1,\dots,g_n)&\doteq
\omega_{0}(\Phi^{\rm int}_{g_1}\star_H\dots\star_H \Phi^{\rm int}_{g_n}) \notag \\
&=(R_{{\lambda} V}(\Phi_{g_1})\star_H\dots\star_H R_{{\lambda} V}(\Phi_{g_1}))(0)\,.
\end{align}
Note that the product we used is the product of the free theory, since, following the philosophy of the interaction picture, interacting fields are constructed within the free field algebra.

Alternatively, we can write the above correlation function as
\begin{align}
\omega_n^{\rm int}(g_1,\dots,g_n)&=\omega_{0}\circ R_{{\lambda} V}\circ R_{{\lambda} V}^{-1} (R_{{\lambda} V}(\Phi_{g_1})\star_H\dots\star_H R_{{\lambda} V}(\Phi_{g_n})) \nonumber\\
&=\omega_{\rm int}(\Phi_{g_1}\starHint\dots \starHint \Phi_{g_n})\,.
\end{align}
Here we did not modify the observables, but changed the product and changed the state. Other natural observables to consider include all the local polynomials \eqref{Local:pol}. 
\begin{rem}
	In the pAQFT setting there are two equivalent ways of treating the interacting theory. On the one hand, one can work with the algebra $\fA^{\hbar,\lambda}_H(\Ccal)=(\Fcal(\Ccal)[[\hbar,\lambda]],\star_H)$ and identify physical observables with elements of this algebra by means of $R_{{\lambda} V}$. For example, take $\Phi_f$, as above. Inside $\fA^{\hbar,\lambda}_H(\Ccal)$ the free quantum observable corresponding to this object is just $\Phi_f$, while the interacting observable is identified as $\Phi_f^{\rm int}= R_{{\lambda} V}(\Phi_f)$. For computing the correlation functions we use the product $\star_H$ and the state $\omega_0$ (given by evaluation at $\phi=0$).
	
	On the other hand, we can model interacting fields using $\fA_H^{\rm int}(\Ccal)= (\Fcal(\Ccal)[[\hbar,{\lambda}]],\starHint)$. In this case, the interacting observable corresponding to $\Phi_f$ is just $\Phi_f$, but for computing the correlation functions we use the product $\star_{H,\rm int}$ and the state $\omega_{\rm int}$.
	
	Note that in one approach we work with complicated observables, but a simple product and a simple state, while in the other approach we have simple observables, but the product and the state become complicated. The crucial difference between the two approaches is how we identify physical objects (e.g. the linear field) with elements of $\Fcal[[\hbar,{\lambda}]]$, which is the underlying vector space in both algebras $\fA^{\hbar,\lambda}_H(\Ccal)$ and $\fA_H^{\rm int}(\Ccal)$.
\end{rem}
\begin{rem}
	Note that in  QFT on causal sets there is a natural UV regularization due to the existence of fundamental length scale. Since the theory is defined on discrete sets, none of the problems that appear in continuum, due to singularities of the Feynman propagator, occur here. Hence there is no need for renormalization. However, one has to be careful when taking the continuum limit, since the UV divergences could again occur, if not taken care of properly. We hope to address this issue in our future work.
\end{rem}

Alternative formulas for $R_{{\lambda} V}$ and $\star_{\rm int}$, in terms of Feynman-like diagrams, have been derived in \cite{HR}. Even though \cite{HR} is formulated for the continuum case, the results are algebraic in nature, so apply also to causal sets (see the Appendix for explicit formulas and more detail).
 In the same work, it has also been shown that
\be
R_{{\lambda} V}(F)\Big|_{\hbar=0}= r_{{\lambda} V}(F)\,,
\ee
and
\be
\frac{1}{i\hbar}[F,G]_{\star_{H,\rm int}}\Big|_{\hbar=0}=\{F,G\}_{{\lambda} V}\,.
\ee
With these formulas at hand, one can now implement any interacting theory in numerical simulations, provided the free theory is known. This opens up perspectives for more examples of interesting causal set field theories, where the influence of adding different interaction terms can be tested and compared with the continuum.

\section{Conclusions and Outlook}
In this paper we have shown how to construct a large class of QFT models on causal sets, using methods of perturbative algebraic quantum field theory (pAQFT). For the purpose of defining the free classical theory (the starting point of our construction) we discussed a number of discretized d'Alembert 
operators and their retarded Green functions. We have also proposed a new ansatz for a class of such discretized wave operators, which uses the notion of \textit{preferred past structure}. The latter is an additional structure augmenting those of a causal set. However, we hope that in our future research we will understand better how to obtain this structure more intrinsically. In particular, we want to determine, using numerical simulations, how much choice there is in typical sprinklings in the definition of a preferred past structure. The element of choice can be removed altogether by defining a discrete d'Alembertian that is the average of $P_\Lambda$ over all possible preferred past structures $\Lambda$. We also hope to be able to generalize the ansatz \eqref{NewdAl} so that the dimension of spacetime itself is not an input (as in \cite{BDD10,DG13,Glaser_2014}), but an emergent quantity. 

The quantization scheme we have proposed works for a very general class of interactions and can be used to test ideas of both causal set theory and pAQFT in new ways. In particular, one can study how approximating the continuum works for interacting theories. One can also investigate how our method of introducing interactions using pAQFT framework relates to the more traditional approach using path integrals. On finite causal sets both approaches can be studied by numerical as well as analytical methods, which is typically not the case in continuum QFT.

\begin{acknowledgments}
	We thank William Cunningham, Fay Dowker, Eli Hawkins, Lisa Glaser, Ian Jubb, Christoph Minz, Rafael Sorkin, Sumati Surya and Stav Zalel for very useful discussions about causal sets! We also thank Ted Jacobson and Lee Smolin for useful comments on the manuscript and for pointing out some relevant references. K.R. would like to acknowledge the financial support of EPSRC (through the grant \verb|EP/P021204/1|) and would like to thank the Perimeter Institute for hospitality and ongoing support. The work of E.D.-H. and N.W. was partly supported by summer studentships from the Department of Mathematics, University of York and E.D.-H. also received support from the 2016 EPSRC DTP funds. 
\end{acknowledgments} 

\appendix
\section{Interacting star product in terms of graphs}
For the convenience of the reader, in this section we summarize the results of \cite{HR} concerning formulas for $\star_{H,\rm int}$ and $R_{{\lambda} V}$ in terms of graphs.
\begin{df}
	Let $\Gcal(n)$ denote the set of directed graphs with $n$ vertices labelled $1,\dots, n$ (and possibly unlabelled vertices with valency $\geq1$) and
	\be
	\Gcal \doteq \bigcup_{n\in\NN} \Gcal(n) \,.
	\ee
	For $\gamma\in\Gcal$: $e(\gamma)$ is the number of edges; $v(\gamma)$ is the number of unlabelled vertices; $\Aut(\gamma)$ is the group of automorphisms.
\end{df}

\begin{df}
	A graph $\gamma\in\Gcal(n)$ determines an $n$-ary multidifferential operator, $\vec\gamma$, on functionals as follows: 
	\begin{itemize}[nosep]
		\item
		An edge represents $E^-(x,y)$ with the direction from $y$ to $x$ --- i.e., such that this is only nonvanishing when the edge points from the future to the past;
		\item
		if the labelled vertex $j$ has valency $r$, this represents the order $r$ functional derivative of the $j$'th argument;
		\item
		likewise, an unlabelled vertex of valency $r$ represents $-V^{(r)}$.
	\end{itemize}
\end{df}

\begin{df}
	$\Gcal_3(n)\subset\Gcal(n)$ is the set of  graphs such that:
	\begin{itemize}[nosep]
		\item
		Every unlabelled vertex has at least one ingoing edge and one outgoing edge;
		\item
		there are no directed cycles;
		\item
		for $1\leq j<k\leq n$, there does not exist any directed path from $j$ to $k$.
	\end{itemize}
\end{df}
In particular, this implies that $1$ is a sink (has only ingoing edges) and $n$ is a source (has only outgoing edges).

Let $\starbint$ be the product on the space of observables $\Fcal(\Ccal)[[\hbar]]$, defined analogously to \eqref{starint}, but where $\T$ is replaced by $\cdot$ and $\star$ is replaced by $\starb$ given by
\be
(F\starb G)(\ph) = \sum_{n=0}^\infty \frac{(-i\hbar)^n}{n!} \left<F^{(n)}(\ph),(E^-)^{\otimes n} G^{(n)}(\ph)\right>\,.
\ee
We can think of it as the version of the interacting star product where the identification between classical and quantum observables (i.e. normal ordering) has been done using the $\Tcal$ map, rather than $\alpha_{H}^{-1}$. With the notation above, in \cite{HR} it was shown that: 
\begin{thm}
	\label{Interacting product}
	\be
	\label{graph star V}
	F\starbint G = \sum_{\gamma\in\Gcal_3(2)}\frac{(-i\hbar)^{e(\gamma)-v(\gamma)}(-\la)^{v(\gamma)}}{\abs{\Aut(\gamma)}} \vec{\gamma}(F,G)
	\ee
\end{thm}
Next, we give the formula for the interacting star product in the more standard formulation, where classical observables are identified with quantum ones by means of normal ordering (i.e. by applying the map $\alpha^{-1}_H$, defined in section \ref{sec:free})
\begin{df}
	$\Gcal_6(n)$ is the set of (isomorphism classes\footnote{The definition of isomorphism classes is a bit more involved technically and not essential for the current application. We refer the reader to \cite{HR} for details.}) of graphs with directed and undirected edges and labelled vertices $1,\dots,n$ such that:
	\begin{itemize}[nosep]
		\item
		Each unlabelled vertex is at least 3-valent, with at least one ingoing and one outgoing edge;
		\item
		there exist no directed cycles;
		\item
		for $1\leq j<k \leq n$, there does not exist a directed path from $j$ to $k$.
	\end{itemize}
\end{df}
\begin{df}
	A graph $\gamma\in \Gcal_6(n)$ defines an $n$-ary multidifferential operator, $\acts{\gamma}$, as follows:
	\begin{itemize}[nosep]
		\item
		A directed edge represents $E_{{\lambda} V}^{-}$;
		\item
		an undirected edge represents $\Delta^{\mathrm F}$;
		\item
		the vertex $j$ represents a derivative of the $j$'th argument;
		\item
		an unlabelled vertex represents a derivative of $-S$, the total action (equivalently, this is just the derivative of $-V$, since unlabelled vertices are at least 3-valent).
	\end{itemize}
\end{df}
\begin{df}
	$\Gcal_7(n)\subset\Gcal_6(n)$ is the subset of graphs with no loops at labelled vertices (i.e., no edge begins and ends at the same labelled vertex).
	$\Gcal_8(n)\subset\Gcal_6(n)$ is the subset of graphs with no loops.
\end{df}

\begin{thm}\label{non:pert:H}
	\be
	\label{HV product}
	F \starHint G = \sum_{\gamma\in \Gcal_7(2)} \frac{(-i)^{v(\gamma)+d(\gamma)}\hbar^{e(\gamma)-v(\gamma)}}{\abs{\Aut\gamma}} \acts{\gamma}(F,G) 
	\ee
	where $d(\gamma)$ is the number of directed edges. In particular, this is a finite sum at each order in $\hbar$.
\end{thm}

Finally, we give also the explicit formulas for the retarded M{\o}ller map itself.
\begin{df}
	$\Gcal_{11}(1)$ is the set of isomorphism classes of graphs with directed and undirected edges and a labelled vertex 1, such that
	\begin{itemize}[nosep]
		\item
		Every unlabelled vertex has at least one incoming edge;
		\item
		1 is a source;
		\item
		there are no directed cycles;
		\item
		there are no loops.
	\end{itemize} 
\end{df}
With this, for the given interaction $V$, the interacting observable corresponding to $F$ is 
\be
R_{\la V}(F) = \sum_{\gamma\in\Gcal_{11}(1)} \frac{(-i)^{d(\gamma)-v(\gamma)}(-\la)^{v(\gamma)}\hbar^{e(\gamma)-v(\gamma)}}{\abs{\Aut(\gamma)}} \vec{\gamma}(F)
\ee
where undirected edges represent $\Delta^{\mathrm F}$, unlabelled vertices correspond to derivatives of $-{\lambda} V$ and $d(\gamma)$ is the number of directed edges. As before, we can also give a non-perturbative (in ${\lambda}$) formula, where we sum up the contributions containing $E^{-}$ to obtain an expression that depends only on the full interacting Green function $E^{-}_{{\lambda} V}$.

\begin{df}
	$\Gcal_{12}(1)\subset \Gcal_{11}(1)$ is the subset of graphs such that no unlabelled vertex has one incoming edge, one outgoing edge, and no unlabelled edge.
\end{df}
Any graph in $\Gcal_{11}(1)$ can be obtained by adding vertices along directed edges of a graph in $\Gcal_{12}(1)$. In this way, the formula for the M{\o}ller map can be re-expressed as
\be
R_{H,\la V}(F) = \sum_{\gamma\in\Gcal_{12}(1)} \frac{(-i)^{v(\gamma)+d(\gamma)}\hbar^{e(\gamma)-v(\gamma)}}{\abs{\Aut(\gamma)}} \acts{\gamma}(F)\,,
\ee
where directed edges represent $E^{-}_{{\lambda} V}$, undirected edges represent $\Delta^{\mathrm F}$, and unlabelled vertices represent derivatives of $-\la V$.

\end{document}